\documentclass[journal]{IEEEtran}

\usepackage{amsmath}
\usepackage{amssymb}
\usepackage{amsthm}
\usepackage{bm}
\usepackage{mathtools}
\usepackage{cite}
\usepackage[ruled]{algorithm2e}
\usepackage{graphicx}
\usepackage{comment}

\DeclareMathOperator{\diag}{diag}
\DeclareMathOperator{\enc}{enc}
\DeclareMathOperator{\sech}{sech}

\newtheorem{theorem}{Theorem}
\newtheorem{problem}{Problem}
\newtheorem{remark}{Remark}
\newtheorem{definition}{Definition}
\newtheorem{example}{Example}
\AtBeginEnvironment{example}{%
  \pushQED{\qed}%
}
\AtEndEnvironment{example}{\popQED\endexample}

\begin{document}

\title{Soft Demapping of Spherical Codes from \\Cartesian Powers of PAM Constellations}

\author{Reza~Rafie~Borujeny~\IEEEmembership{Member,~IEEE}, Susanna~E.~Rumsey~\IEEEmembership{Graduate~Student~Member,~IEEE}, Stark~C.~Draper~\IEEEmembership{Senior~Member,~IEEE}, Frank~R.~Kschischang~\IEEEmembership{Fellow,~IEEE}%
\thanks{The authors are with the Edward S. Rogers Sr. Department of Electrical and Computer Engineering, University of Toronto, Toronto, ON M5S 3G4, Canada (e-mail: rrafie@ece.utoronto.ca, s.rumsey@mail.utoronto.ca, stark.draper@utoronto.ca, frank@ece.utoronto.ca).}}

\maketitle

\begin{abstract}
For applications in concatenated coding for optical communications systems, we examine soft-demapping of short spherical codes constructed as constant-energy shells of the Cartesian power of pulse amplitude modulation constellations.  These are unions of permutation codes having the same average power.  We construct a list decoder for permutation codes by adapting Murty's algorithm, which is then used to determine mutual information curves for these permutation codes. In the process, we discover a straightforward expression for determining the likelihood of large subcodes of permutation codes. We refer to these subcodes, obtained by all possible sign flips of a given permutation codeword, as orbits. We introduce a simple process, which we call orbit demapping with frozen symbols, that allows us to extract soft information from noisy permutation codewords. In a sample communication system with probabilistic amplitude shaping protected by a standard low-density parity-check code that employs short permutation codes, we demonstrate that orbit demapping with frozen symbols provides a gain of about $0.3$~dB in signal-to-noise ratio compared to the traditional symbol-by-symbol demapping. By using spherical codes composed of unions of permutation codes, we can increase the input entropy compared to using permutation codes alone. In one scheme, we consider a union of a small number of permutation codes. In this case, orbit demapping with frozen symbols provides about $0.2$~dB gain compared to the traditional method. In another scheme, we use all possible permutations to form a spherical code that exhibits a computationally feasible trellis representation. The soft information obtained using the BCJR algorithm outperforms the traditional symbol-by-symbol method by $0.1$~dB. Overall, using the spherical codes containing all possible permutation codes of the same average power and the BCJR algorithm, a gain of $0.5$~dB is observed compared with the case of using one permutation code with the symbol-by-symbol demapping. Comparison of the achievable information rates of bit-metric decoding verifies the observed gains.
\end{abstract}

\begin{IEEEkeywords}
    Probabilistic amplitude shaping, spherical codes, permutation codes, trellis, soft-out demapping.
\end{IEEEkeywords}

\section{Introduction}
\IEEEPARstart{R}{ecent} advancements in optical fiber communication have highlighted the potential of employing short spherical codes for modulation, leading to promising nonlinear gains~\cite{rafie_2023_why}. These spherical codes, constructed from Cartesian powers of pulse amplitude modulation (PAM) constellations, are essentially unions of permutation codes~\cite{slepian_1965_permutation}. Permutation codes have already found their application in probabilistic amplitude shaping (PAS) within the realm of fiber optic communication~\cite{bocherer_2015_bandwidth}. This paper aims to delve deeper into the study of such unions of permutation codes, exploring their potential benefits and implications in the context of optical fiber communications.

Permutation codes manifest three intriguing properties. Firstly, all codewords possess the same energy. This characteristic has been identified as the cause for their reduced cross-phase modulation (XPM) and enhanced received signal-to-noise ratio (SNR) when deployed over a nonlinear optical fiber~\cite{rafie_2023_why}, serving as the primary motivation for this work.

Secondly, all codewords, when considering the amplitude of their elements, exhibit the same type. Consequently, they can be utilized to approximate a probability mass function on their constituent constellation (i.e., the alphabet from which the elements of the initial vector are selected). This capability enables their use in distribution matching~\cite{bocherer_2016_constant}, a key property implicitly fundamental to the development of PAS~\cite{bocherer_2015_bandwidth}. 

Thirdly, minimum Euclidean distance decoding of permutation codes is notably straightforward~\cite{slepian_1965_permutation}, with a decoding complexity not worse than that of sorting.

The performance of permutation codes has been extensively studied in terms of their distance properties and error probability under maximum likelihood detection when transmitted over an additive white Gaussian noise (AWGN) channel (see~\cite{ingemarsson_1989_group} and references therein). The performance of permutation codes in terms of mutual information in an AWGN channel has remained largely elusive until recent advancements. In particular, the work of~\cite{bocherer_2015_bandwidth} suggests that, given reasonable protection with a forward error-correcting code, subsets of permutation codes with large blocklength can operate remarkably close to the Shannon limit.

To circumvent the inherent rate loss at short blocklengths~\cite{gultekin_2020_probabilistic}, permutation codes of sufficiently large blocklength should be employed for PAS~\cite{bocherer_2023_probabilistic}. Conversely, to leverage their enhanced SNR in nonlinear fiber, a smaller blocklength is preferred~\cite{rafie_2023_why}. This presents a trade-off between the blocklength and the linear/nonlinear gains, which is an area of focus in this work. Our primary interest lies in the application of permutation codes, and generalizations thereof, as inner multi-dimensional modulation formats in a concatenated scheme with an outer error-correcting code\footnote{In the reverse concatenation scheme used in PAS, the meanings of ``inner'' and ``outer'' should be interpreted from the point of view of the receiver.}. We consider the PAS~\cite{bocherer_2015_bandwidth} setup to demonstrate our findings, as detailed in Section~\ref{sec:background}. The novelty of our work lies in our choice of signal set for the inner modulation and in the soft demapping of the inner modulation.

We will focus on permutation codes with a PAM constituent constellation. This maintains standard transmitter designs, making our codes compatible with existing systems. As mentioned, the performance of permutation codes in terms of mutual information has not been extensively studied in the literature. This is chiefly due to the very large size of these codes, which makes conventional computation of mutual information practically impossible. In Sections~\ref{sec:mivssnr} and~\ref{sec:softdecode}, different types of achievable rates of permutation codes are discussed. In particular, using a novel list decoder for permutation codes constructed as a special case of Murty's algorithm~\cite{murty_1968_algorithm}, we develop a methodology to obtain mutual information curves of permutation codes for the AWGN channel---a task that was otherwise infeasible using conventional methods. These curves, which are counterparts of the mutual information curves of~\cite{ungerboeck_1982_channel} for PAM constellations, characterize the mutual information as a function of SNR for permutation codes when, at the transmitter, codewords are chosen uniformly at random. Using these curves, we study the shaping performance of permutation codes as a function of their blocklength and find the range of blocklengths for which a reasonable shaping gain can be expected. 

The soft-demapping of permutation codes (or the more general spherical codes we will study) has received little attention in the literature. One common approach to obtaining soft information is to assume that the elements of each codeword are independent~\cite{bocherer_2015_bandwidth}. This significantly reduces the computational complexity of calculating the soft information, as it can be computed in a symbol-by-symbol fashion. The type of the permutation code contributes to the calculation of the soft information through the prior distribution. While the soft information obtained using this assumption becomes more accurate as the blocklength goes to infinity, for the small blocklengths of interest in nonlinear fiber, the symbol-by-symbol soft-demapping approach of~\cite{bocherer_2015_bandwidth} is far from accurate~\cite{schulte_2020_joint, luo_2023_joint}. One key result of this work is a simple and computationally efficient method, which is described in Section~\ref{sec:softdecode}, to obtain soft information from received permutation codewords. 

Finding the largest permutation code for a given average energy and a fixed constituent constellation is a finite-dimensional counterpart of finding the maximizer of input entropy with an average power constraint~\cite{kschischang_1993_optimal}. By removing the constraint on having a permutation code (i.e., having only one type), the problem of finding the largest spherical code for a given average energy and a fixed constituent constellation is more relevant at finite blocklengths. This will allow us to further increase the input entropy at shorter blocklengths of interest in communication over nonlinear optical fiber while preserving the constant-energy property. Most importantly, in Section~\ref{sec:shell}, we show that such codes exhibit a trellis structure with reasonable trellis complexity, allowing us to obtain soft information from the received spherical codewords. Moreover, the achievable rates of permutation codes under bit-metric decoding are numerically evaluated using various methods for estimating soft information. The soft-in soft-out structure of our scheme makes it an appealing choice for fiber-optic communications systems, as it allows for short constant-energy spherical codes with reasonable shaping gain and tolerance to nonlinear impairments caused by XPM. 

Throughout this paper, we adopt the following notational conventions to ensure clarity and consistency.
The set of real numbers is denoted as $\mathbb{R}$.
A normal lowercase italicized font is used for variables, samples of a random variable and functions with one-letter names, as in $x, y, f(x)$ or $\alpha(x)$.
A bold lowercase italicized font is used for vectors, which are always treated as column vectors, as in $\bm{x} = (x_1,x_2, \dots, x_n)$.
A normal capital italicized font is used for random variables and graphs, as in $X, Y, Z$ and $G$. The only exception to this rule is $O$ which is used as Landau's big-O symbol.
A bold capital italicized font is used for random vectors, as in $\bm{X}$.
A sans-serif capital font is used for matrices, as in $\mathsf{P}$.
A capital calligraphy font is used for sets and groups, as in $\mathcal{A}, \mathcal{S}$ and $\mathcal{C}$.
A monospace font is used for energy (squared euclidean norm) of vectors as in $\mathtt{E}$ or $\mathtt{E}_{\bm{y}}$.

\section{Background}\label{sec:background}
This section contains an overview of some of the most important concepts used in this paper including a brief description of permutation codes and a general overview of PAS. We also review the assignment problem, a fundamental combinatorial optimization problem which will play a central role in the development of list decoders for permutation codes.

\subsection{Permutation Codes}
We denote the set of the first $n$ positive integers by $[n]$. A \emph{permutation} of length $n$ is a bijection from $[n]$ to itself. The set of all permutations of length $n$ forms a group under function composition. This group is called the \emph{symmetric group} on $[n]$ and is denoted by $\mathcal{S}_n$. 

There are different ways to represent a permutation $\alpha\in \mathcal{S}_n$. In the \emph{two-line form}, the permutation $\alpha$ is given as
\[
\phantom{.}\begin{bmatrix}
1 & 2 & 3 & \hdots & n\\
\alpha(1) & \alpha(2) & \alpha(3) & \hdots & \alpha(n)
\end{bmatrix}.
\]
The \emph{matrix form} of $\alpha$ is given as
\[
\mathsf{P}_\alpha = \begin{bmatrix}
\bm{e}_{\alpha(1)} & \bm{e}_{\alpha(2)} & \bm{e}_{\alpha(3)} & \hdots & \bm{e}_{\alpha(n)}
\end{bmatrix}^\top
\]
where $\bm{e}_{j}$ denotes a column vector of length $n$ with $1$ in the $j^{\text{th}}$ position and $0$ in every other position. 
The set of permutation matrices $\{\mathsf{P}_\alpha \mid \alpha \in \mathcal{S}_n\}$ forms a group under matrix multiplication and is usually denoted as $\mathcal{A}_{n-1}$. This group is, of course, isomorphic to $\mathcal{S}_n$. A \emph{signed permutation matrix} is a generalized permutation matrix whose nonzero entries are $\pm 1$. The set of all $n\times n$ signed permutation matrices forms a group under matrix multiplication and is usually denoted as $\mathcal{B}_n$.

A \emph{Variant~I} permutation code $\mathcal{C}$ of length $n$~\cite{slepian_1965_permutation} is characterized as the set of points in real Euclidean $n$-space $\mathbb{R}^n$ that are generated by the group $\mathcal{A}_{n-1}$ acting on a vector $\bm{x}$:
\[
\phantom{.}\mathcal{C} = \{\mathsf{P}\bm{x} \mid \mathsf{P}\in \mathcal{A}_{n-1}\}.
\]
This means that the code $\mathcal{C}$ is the set of all distinct vectors that can be formed by permuting the order of the $n$ elements of a vector $\bm{x}$. Each vector in $\mathcal{C}$ is called a codeword.
Likewise, a \emph{Variant~II} permutation code $\mathcal{C}$ of length $n$ is characterized as the set of points in $\mathbb{R}^n$ that are generated by the group $\mathcal{B}_n$ acting on a vector $\bm{x}$:
\[
\phantom{.}\mathcal{C} = \{\mathsf{P}\bm{x} \mid \mathsf{P}\in \mathcal{B}_{n}\}.
\]
Here, the code $\mathcal{C}$ is the set of all distinct vectors that can be formed by permuting the order or changing the sign of the $n$ elements of a vector $\bm{x}$.

The \emph{initial vector} of a Variant~I or II permutation code is the codeword 
\[
\bm{x} = (x_1, x_2, x_3, \dots, x_n)
\]
that satisfies $0 < x_1\leq x_2 \leq x_3 \leq \dots \leq x_n$. The left most inequality is strict to be consistent with the use of a PAM constituent constellation.

%%%%%%%%%%%%%%%%%%%%%%%%%%%%%%%%%%%%%%%%%%%
\begin{figure}
\centering
\includegraphics[width=\columnwidth]{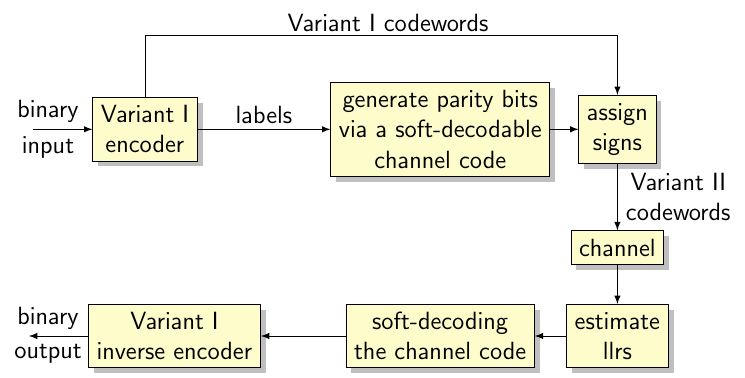}
\caption{Block diagram of the communication system that we consider.}
\label{fig:pas}
\end{figure}
%%%%%%%%%%%%%%%%%%%%%%%%%%%%%%%%%%%%%%%%%%%

\subsection{Probabilistic Amplitude Shaping}
\label{sec:PAS}
We consider the communication system of PAS as described in~\cite{bocherer_2015_bandwidth}. As shown in Fig.~\ref{fig:pas}, at the first stage in the transmitter, a sequence of binary information inputs is mapped to a Variant~I permutation code whose initial vector has an empirical distribution that resembles a desired input distribution. This first step is called \emph{permutation encoding} and is described further in Section~\ref{sec:softdecode}. Each symbol in the constituent constellation is \emph{labeled} according to some Gray labeling of the used PAM. Therefore, corresponding to each Variant~I codeword, there is a \emph{binary label} obtained by concatenating the Gray label of each of the elements in the codeword. If the constituent constellation\footnote{We assume that the constituent constellation has a power-of-two size.} is $\mathcal{M}$, the length of the label of each Variant~I codeword of blocklength $n$ is exactly $n\log_2\lvert \mathcal{M}\rvert$. The binary label of a codeword should not be confused with the information bits at the input of the permutation encoder that produced the codeword. Each Variant~I permutation codeword is the output of the permutation encoder applied to a sequence of typically less than $n\log_2\lvert \mathcal{M}\rvert$ input information bits. 

To protect the labels, a soft-decodable channel code is used. We emphasize that instead of encoding the input information bits that was used by the permutation encoder, the codeword labels are encoded by the channel code. This is fundamental in the development of PAS that uses bit-interleaved coded modulation as it allows the permutation codewords to have a semi-Gray labeling: if codewords are close in Euclidean distance, their labels are also close in the Hamming distance. If the input information bits are used to label the codewords, a semi-Gray labeling, even for small blocklengths, appears to be challenging to find. The channel encoder produces a sequence of parity bits which is used to assign signs to the elements of the Variant~I permutation codeword\footnote{Sometimes, some binary inputs are directly used as sign bits of some of the elements of the permutation codeword. This is done when the ratio of the blocklength of the permutation code to the length of its binary label, $1/\log_2\lvert\mathcal{M}\rvert$, is larger than the overhead of the channel code. See~\cite{bocherer_2015_bandwidth} for details.}. This will produce a Variant~II permutation codeword which will be transmitted over an AWGN channel with noise variance $\sigma^2$. 

A noisy version of each transmitted Variant~II permutation codeword is received at the receiver. By \emph{demapping}, we mean the detection of the binary label and sign bits of a permutation codeword from its received noisy version. The first block at the receiver estimates the log-likelihood ratio (LLR) for each bit in the label of the received word---a process we refer to as \emph{soft demapping}. The estimated LLRs are then passed to the decoder of the channel code to approximate the label of the transmitted Variant~I permutation codeword at the receiver. If an acceptable Variant~I codeword is obtained, the inverse of permutation encoder produces the corresponding binary input sequence. Conversely, if the reproduced symbols do not form a valid Variant~I permutation codeword, a block error is declared.

\subsection{The Assignment Problem}\label{subsec:assignmentproblem}
A bipartite graph $G = (\mathcal{U},\mathcal{V},\mathcal{E})$ is a triple consisting of two disjoint set of vertices $\mathcal{U}$ and $\mathcal{V}$ and a set of edges $\mathcal{E}\subseteq \mathcal{U}\times \mathcal{V}$. Each edge in $\mathcal{E}$ connects a vertex in $\mathcal{U}$ to a vertex in $\mathcal{V}$. The graph $G$ is called balanced if $\lvert \mathcal{U}\rvert = \lvert \mathcal{V}\rvert = n$. The bipartite graph $G$ is complete if $\mathcal{E} = \mathcal{U}\times \mathcal{V}$. All bipartite graphs that we consider are assumed complete and balanced. The graph $G$ is weighted if there exist a weight function $w:\mathcal{E}\to\mathbb{R}\cup\{-\infty\}$. A matching is a set of edges without common vertices. The weight of a matching is the sum of the weight of the edges in that matching. A perfect matching is a matching that matches all vertices of the graph. The maximum weight perfect matching problem is the problem of finding a perfect matching with maximum weight. This problem is usually called the assignment problem and is stated as follows. 
\begin{problem}[Assignment Problem]
Find a bijection $f:\mathcal{U}\to \mathcal{V}$ such that the reward function 
\[
\sum_{a\in \mathcal{U}}w(a,f(a))
\]
is maximized.
\end{problem}

The assignment problem is typically solved using the Hungarian method~\cite{burkard_2012_assignment, kuhn_1955_hungarian}. The time complexity of the Hungarian algorithm is $O\left(n^3\right)$. However, the weight function may have additional structure that can be exploited to reduce the complexity of solving the assignment problem. 

A weight function $w$ is said to be \emph{multiplicative} if there exist two functions
\begin{align*}
w_{\mathcal{U}}:\mathcal{U}\to\mathbb{R}\cup\{-\infty\},\\
w_{\mathcal{V}}:\mathcal{V}\to\mathbb{R}\cup\{-\infty\},
\end{align*}
such that 
\[
\phantom{.}\forall (a,b) \in \mathcal{U}\times \mathcal{V} : w(a,b)=w_{\mathcal{U}}(a)w_{\mathcal{V}}(b).
\]
If the weight function is multiplicative the assignment problem can be solved using sorting as follows: Assume
\[
\phantom{,}\mathcal{U} = \{a_i\mid i\in [n]\},\quad \mathcal{V} = \{b_i\mid i\in [n]\},
\]
and that there are two permutations $\sigma_{\mathcal{U}}$ and $\sigma_{\mathcal{V}}$ on $[n]$ such that
\begin{align*}
w_{\mathcal{U}}(a_{\sigma_{\mathcal{U}}(1)})\leq w_{\mathcal{U}}(a_{\sigma_{\mathcal{U}}(2)}) \leq \dots \leq w_{\mathcal{U}}(a_{\sigma_{\mathcal{U}}(n)}),\\
w_{\mathcal{V}}(b_{\sigma_{\mathcal{V}}(1)})\leq w_{\mathcal{V}}(b_{\sigma_{\mathcal{V}}(2)}) \leq \dots \leq w_{\mathcal{V}}(b_{\sigma_{\mathcal{V}}(n)}).
\end{align*}
In this case, the assignment problem can be immediately solved using the upper bound in the rearrangement inequality~\cite{hardy_1988_inequalities}.
The bijection $f$ that solves the assignment problem is defined by
\[
\phantom{.}f(a_i) = b_{\sigma_{\mathcal{V}}\sigma_{\mathcal{U}}^{-1}(i)}.
\]
The solution can be found by sorting in $O(n\log n)$ time.

\section{Permutation Codes as Multidimensional Constellations}\label{sec:mivssnr}
In this section, we discuss how permutation codes can be used as multidimensional constellations. List decoding of permutation codes is formulated as an assignment problem and proper weight functions for permutation codes are defined. In particular, a novel weight function for Variant~II permutation codes is introduced that allows for efficient \emph{orbit list decoding} of Variant~II permutation codes. Orbit decoding is used in the rest of the paper in order to obtain soft information when demapping Variant~II permutation codes.
In the Appendix, using the definitions of this section, we present a methodology to obtain mutual information vs. SNR curves for permutation codes. Essential to the development of our method is a specialized version of Murty's algorithm which is described in the Appendix. This section concludes with example mutual information curves which show that considerable shaping gain is obtained by using permutation codes with blocklengths as small as $n = 50$.

\subsection{List Decoding of Variant~I Permutation Codes}
We consider a vector Gaussian noise channel 
\[
\bm{Y} = \bm{X} + \sigma \bm{Z}
\]
where the input $\bm{X}$ is a random vector uniformly distributed over a permutation code $\mathcal{C}$, the noise term $\bm{Z}$ is a standard normal random vector and $\sigma$ is the standard deviation of the additive white Gaussian noise. The output $\bm{Y}$ is a noisy version of $\bm{X}$.
The joint distribution\footnote{Instead of distinguishing between discrete, continuous or mixed random vectors, we always refer to $p$ as a distribution for convenience.} of the input and the output is denoted by $p(\bm{x},\bm{y})$. Similarly, $p(\bm{y}\mid \bm{x})$ is the conditional distribution of the output conditional on the input and $p(\bm{y})$ is the output distribution. List decoding of a Variant~I permutation code is to find a set $\mathcal{L} = \{\bm{c}_0, \bm{c}_1, \dots, \bm{c}_{L-1}\}$ of codewords with the highest likelihoods so that
\[
\phantom{,}p(\bm{y}\mid \bm{c}_0) \geq p(\bm{y}\mid \bm{c}_1) \geq \dots \geq p(\bm{y}\mid \bm{c}_{L-1}),
\]
where $L = \lvert\mathcal{L}\rvert$ is the \emph{list size}.

To be able to find such a list of candidate codewords, we need to define a weight function for the maximum likelihood decoding of permutation codes. Because the noise is Gaussian, a codeword $\bm{x}$ has a higher likelihood than a codeword $\bm{x}'$ if $\bm{x}$ is closer to the received word $\bm{y}$ than $\bm{x}'$. That is,
\[
\phantom{.}p(\bm{y}\mid \bm{x}) \geq p(\bm{y}\mid \bm{x}') \iff \sum_{i=1}^n(y_i-x_i)^2 \leq \sum_{i=1}^n(y_i-x_i')^2. 
\]
Using the constant-energy property of permutation codes, we can simplify further to get
\[
\phantom{.}p(\bm{y}\mid \bm{x}) \geq p(\bm{y}\mid \bm{x}') \iff \sum_{i=1}^n y_ix_i \geq \sum_{i=1}^ny_ix_i'.
\]
Therefore, list decoding with maximum likelihood reward is equivalent to list decoding with maximum correlation reward. For Variant~I permutation codes, we can define the following assignment problem:

\begin{problem}[Maximum Correlation Decoding]\label{prob:assignment_correlation}
Find a bijection $f:[n]\to [n]$ such that the reward function
$
\sum_{i=1}^n x_i y_{f(i)}
$
is maximized.
\end{problem}
This is an assignment problem with
\[
\phantom{,}\mathcal{U} = \{x_i\mid i\in[n]\},\quad \mathcal{V} = \{y_i\mid i\in[n]\},
\]
and weight function
\[
\phantom{.}w(x_i,y_j) = x_iy_j.
\]
The weight function is clearly multiplicative with $w_{\mathcal{U}}(x_i) = x_i$ and $w_{\mathcal{V}}(y_j) = y_j$. The list decoding problem, then, can be solved by a direct application of Murty's algorithm as described in the Appendix. 

The assignment problem defined for Variant~I permutation codes is not suitable for Variant~II permutation codes as in a Variant~II code, each symbol can have a plus or minus sign. It is not possible to accommodate such arbitrary sign flips in the assignment problem defined for Variant~I permutation codes. In what follows, we introduce maximum orbit-likelihood decoding of Variant~II permutation codes that allows us to reformulate the decoding problem in the framework of an assignment problem.

\subsection{Orbit Decoding of Variant~II Permutation Codes}
A trivial subcode of a Variant~II permutation code $\mathcal{C}$ corresponding to a codeword $\bm{x}\in \mathcal{C}$ is the subcode containing $\bm{x}$ and all codewords obtained by changing the sign of elements in $\bm{x}$ in all possible ways. We denote the group of $n$-dimensional diagonal unitary matrices (i.e., $n \times n$ diagonal integer matrices with $\pm 1$ diagonal entries) as $\mathcal{D}_n$. A trivial subcode corresponding to $\bm{x}$ is the orbit of $\bm{x}$ under the group action $\phi$ defined by
\begin{align}
\phi:\mathcal{D}_n\times \mathcal{C} &\rightarrow \mathcal{C},\nonumber\\
(\mathsf{A}, \bm{x}) &\mapsto \mathsf{A}\bm{x}.\nonumber
\end{align}
Similar to finding the codeword in a Variant~I permutation code with the maximum likelihood, we can think about finding the orbit in a Variant~II permutation code with the maximum orbit likelihood---a process that we call maximum orbit-likelihood decoding. Also, similar to a list of $L$ most likely codewords for a Variant~I permutation code, we can think about finding a list of $L$ most likely orbits for a Variant~II permutation code---a process that we refer to as orbit list-decoding. In order to use Murty's algorithm, a weight function for the orbits needs to be defined. We denote the likelihood of the orbit of $\bm{x}\in \mathcal{C}$ by
\[
\phantom{.}f_{\mathrm{o}}(\bm{y}\mid \bm{x})\coloneqq \frac{\sum_{\mathsf{A}\in \mathcal{D}_n} p(\bm{y}\mid \mathsf{A}\bm{x})p(\mathsf{A}\bm{x})}{\sum_{\mathsf{A}\in \mathcal{D}_n} p(\mathsf{A}\bm{x})} = \!\!\!\!\sum_{\mathsf{A}\in \mathcal{D}_n}\!\!\!\!\frac{p(\bm{y}\mid \mathsf{A}\bm{x})}{2^n}.
\]
The key result of this section is the following theorem.
%%%%%%%%%%%%%%%%%%%%%%%%%%%%%%%%%%%%%%%%%%%
\begin{figure}
\centering
\includegraphics[width=\columnwidth]{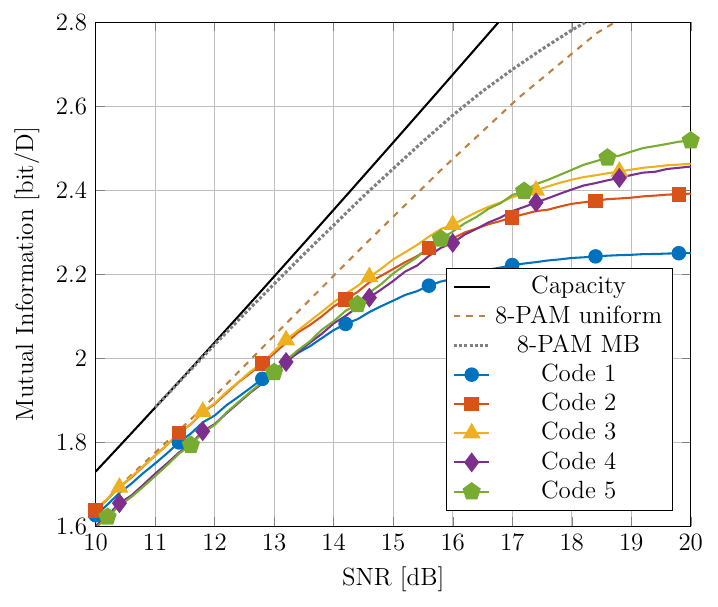}
\caption{Mutual information curves for five codes with $n=12$.}
\label{fig:mi11}
\end{figure}
%%%%%%%%%%%%%%%%%%%%%%%%%%%%%%%%%%%%%%%%%%%
\begin{theorem}\label{thm:orbitlikelihood}
The likelihood of the orbit of $\bm{x}\in \mathcal{C}$ under the action of $\phi$ is given by
\begin{equation}\label{eq:orbitlikelihood}
\phantom{,}f_{\mathrm{o}}(\bm{y}\mid \bm{x}) = a(\mathtt{E}_{\bm{y}}, \mathtt{E}, \sigma, n)\prod_{i=1}^n\cosh\frac{y_ix_i}{\sigma^2},
\end{equation}
where $a(\mathtt{E}_{\bm{y}}, \mathtt{E}, \sigma, n)$ is a constant that depends on the energy of the received word $\mathtt{E}_{\bm{y}}$, the energy $\mathtt{E}$ of each permutation codeword, the noise standard deviation $\sigma$ and the blocklength $n$.
\end{theorem}
\begin{proof}
\begin{align*}
f_{\mathrm{o}}(\bm{y}\mid \bm{x}) &= \sum_{\mathsf{A}\in \mathcal{D}_n}\frac{p(\bm{y}\mid \mathsf{A}\bm{x})}{2^n}\\
&= \frac{1}{\left(2\sqrt{2\pi}\sigma\right)^n}\sum_{s_1=\pm1}\sum_{s_2=\pm1}\!\dots\!\sum_{s_n=\pm1}\prod_{i=1}^ne^{-\frac{(y_i - s_ix_i)^2}{2\sigma^2}}\\
&= \frac{1}{\left(2\sqrt{2\pi}\sigma\right)^n}\prod_{i=1}^n\sum_{s_i=\pm1}e^{-\frac{(y_i - s_ix_i)^2}{2\sigma^2}}\\
&=\frac{1}{\left(2\sqrt{2\pi}\sigma\right)^n}\prod_{i=1}^n\left(e^{-\frac{(y_i - x_i)^2}{2\sigma^2}} + e^{-\frac{(y_i + x_i)^2}{2\sigma^2}}\right)\\
&=\frac{1}{\left(2\sqrt{2\pi}\sigma\right)^n}\prod_{i=1}^ne^{-\frac{y_i^2 + x_i^2}{2\sigma^2}}\left(e^{\frac{y_ix_i}{\sigma^2}}+e^{-\frac{y_ix_i}{\sigma^2}}\right)\\
&=\frac{1}{\left(2\sqrt{2\pi}\sigma\right)^n}\prod_{i=1}^n2e^{-\frac{y_i^2 + x_i^2}{2\sigma^2}}\cosh\frac{y_ix_i}{\sigma^2}\\
&=\frac{e^{-\frac{\mathtt{E}_{\bm{y}}+\mathtt{E}}{2\sigma^2}}}{\left(\sqrt{2\pi}\sigma\right)^n}\prod_{i=1}^n\cosh\frac{y_ix_i}{\sigma^2}.
\end{align*}
\end{proof}
Theorem~\ref{thm:orbitlikelihood} suggests the weight function
\[w(x_i,y_j) = \log\cosh\frac{y_jx_i}{\sigma^2}\]
for Murty's algorithm in orbit list-decoding of Variant~II permutation codes. This weight function, however, is not multiplicative. Using a generalized version of the rearrangement inequality~\cite{holstermann_2017_generalization}, it turns out that this weight function has all the required properties of a multiplicative weight function that we used to simplify Murty's algorithm in the Appendix.

Mutual information of Variant~II permutation codes can be computed using the orbit likelihoods by the method discussed in the Appendix. Note that the contribution of the likelihood of one orbit in $p(\bm{y})$ in~(\ref{eq:approxpy}) is equivalent to the contribution of all $2^n$ codewords in the orbit. 
\begin{example}\label{ex:mi11}
As an example, we find the mutual information vs. SNR for five codes with blocklength $n=12$ with the following initial vectors:
\begin{itemize}
    \item for Code 1, $(1, 1, 1, 1, 1, 3, 3, 3, 3, 3, 5, 7)$,
    \item for Code 2, $(1, 1, 1, 1, 1, 3, 3, 3, 5, 5, 5, 7)$,
    \item for Code 3, $(1, 1, 1, 1, 3, 3, 3, 3, 5, 5, 7, 7)$,
    \item for Code 4, $(1, 1, 1, 1, 3, 3, 5, 5, 5, 5, 7, 7)$,
    \item for Code 5, $(1, 1, 1, 3, 3, 3, 5, 5, 5, 7, 7, 7)$.
\end{itemize}
The average symbol power of each of the codes is $10.33$, $13$, $15.66$, $18.33$ and $21$, respectively, and the initial vectors have been chosen to maximize the code size given these power constraints. Note that the constituent constellation is an $8$-PAM. The average power of this constituent constellation with uniform input distribution is $21$. In Fig.~\ref{fig:mi11}, we illustrate the mutual information curves for these five codes as well as the constituent constellation under uniform input distribution. Also shown is the mutual information of the constituent constellation with Maxwell--Boltzman (MB) distributions which are known to be entropy maximizers~\cite{kschischang_1993_optimal} and close to optimal, in terms of mutual information, for the AWGN channel~\cite{delsad_2023_probabilistic}. The channel capacity is also depicted. It is obvious that none of the five codes show any shaping gain as all mutual information curves are below that of the $8$-PAM with uniform distribution.
\end{example}
%%%%%%%%%%%%%%%%%%%%%%%%%%%%%%%%%%%%%%%%%%%
\begin{figure}
\centering
\includegraphics[width=\columnwidth]{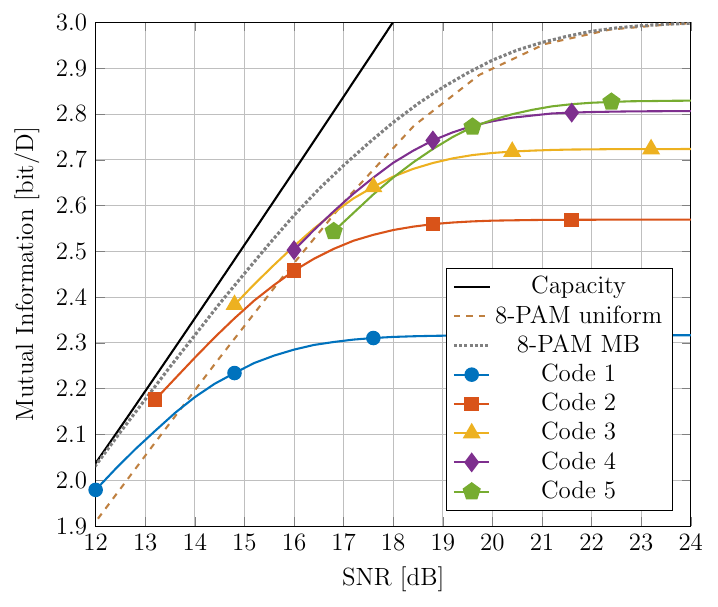}
\caption{Mutual information curves for five codes with $n=50$.}
\label{fig:mi12}
\end{figure}
%%%%%%%%%%%%%%%%%%%%%%%%%%%%%%%%%%%%%%%%%%%
\begin{example}\label{ex:mi12}
As another example, we find the mutual information vs. SNR for five codes with blocklength $n=50$. The initial vectors are chosen so that the largest Variant~II codes with the average power of $7.08$, $10.6$, $14.12$, $17.64$ and $21.16$, are chosen. The list size for Murty's algorithm was set to $5^6$. The number of Monte Carlo samples for each SNR is $100,000$. However, no numerically significant difference was observed even with $10,000$ samples. In Fig.~\ref{fig:mi12}, we illustrate the mutual information curves for these five codes as well as the constituent constellation under uniform and MB input distribution. It is evident that Code 1, Code 2 and Code 3 show considerable shaping gains as their curves cross that of the $8$-PAM with uniform distribution. For lower SNR values, one needs to have larger list sizes to get an accurate estimate for $p(\bm{y})$. As a result, the computation of curves for lower SNRs is slower. We did not generate the curves for lower SNRs as we were mostly interested in the crossing of the permutation codes' mutual information curves and the mutual information curve of the $8$-PAM with uniform distribution. 
\end{example}

\begin{remark}
As shown in Fig.~\ref{fig:mi12}, a short blocklength of $n=50$ is enough to have considerable shaping gain over the AWGN channel. 
In practice, we are usually interested in using binary channel codes and, therefore, the so-called bit-metric decoding (BMD) of permutation codes, in which the mutual information induced via the labels of permutation codewords is considered. In the next section, we present BMD rates for Code 2 of Example \ref{ex:mi12} and its generalizations for various methods of soft demapping.
\end{remark}

\section{Soft Demapping of Permutation Codes}\label{sec:softdecode}
\subsection{Variant~I Encoder Function and its Inverse}\label{subsec:varienc}

\begin{algorithm}
\caption{An encoding algorithm based on~\cite{nordio_2003_permutation}. The output of the algorithm is a vector of indices $\bm{c}$ of length~$n$. Each index is an integer between $1$ and $u$ and represents a symbol in the vector $\bm{\mu}$. That is, the corresponding codeword is $\bm{\mu}[\bm{c}]$.}\label{alg:enc_general}
\DontPrintSemicolon
\KwIn{codeword integer index $q$, vector of repetitions $\bm{m}$, cardinality of the code $\mathit{size}$, code length $n$}
\KwOut{index vector $\bm{c}$}
$\bm{c} \leftarrow$ initialize to a vector of length $n$\;
\For{$\mathit{index} \leftarrow 1$ \KwTo $n$}{
  $\mathit{letter} \leftarrow 0$\;
  $\mathit{size}_{\mathrm{temp}} \leftarrow 0$\;
  \While{$q \geq 0$}{
    $\mathit{letter} \leftarrow \mathit{letter} + 1$\;
    $\mathit{size}_{\mathrm{temp}} \leftarrow \lfloor(\mathit{size} \cdot \bm{m}[\mathit{letter}]) / (n - \mathit{index} + 1)\rfloor$\;
    $q \leftarrow q - \mathit{size}_{\mathrm{temp}}$\;
  }
  $\mathit{size} \leftarrow \mathit{size}_{\mathrm{temp}}$\;
  $q \leftarrow q + \mathit{size}_{\mathrm{temp}}$\;
  $\bm{c}[\mathit{index}] \leftarrow \mathit{letter}$\;
  $\bm{m}[\mathit{letter}] \leftarrow \bm{m}[\mathit{letter}] - 1$\;
}
\Return{$\bm{c}$}\;
\end{algorithm}

A Variant~I permutation code $\mathcal{C}$ is uniquely identified by its initial vector $\bm{x}$. Note that when the length of the code is $n$, the number of distinct elements in $\bm{x}$ is at most $n$. That is,
\[
\phantom{.}\lvert\{x_j \mid j\in [n] \}\rvert \leq n.
\]
We denote the number of unique elements of $\bm{x}$ by $u$. We define $\bm{\mu}$ as the vector whose elements are the distinct elements of $\bm{x}$ in ascending order. That is, if
\[
\phantom{,}\bm{\mu} = (\mu_1, \mu_2, \mu_2, \dots, \mu_u),
\]
then $\mu_1 < \mu_2 < \mu_3 < \dots < \mu_u$. The number of times each $\mu_i$ occurs in $\bm{x}$ is denoted as $m_i$. This means that
\[
\phantom{.}m_1 + m_2 + m_3 + \dots + m_u = n.
\]
The size of this code is
\begin{equation}\nonumber%\label{eq:permSize}
\lvert \mathcal{C}\rvert = \frac{n!}{m_1!m_2!m_3!\dots m_u!}
\end{equation}
and its rate is $r = \frac{\log_2\lvert \mathcal{C}\rvert}{n}$.

The problem of encoding permutation codes has been considered by many~\cite{berger_1972_permutation, bocherer_2016_constant, myrvold_2001_ranking, nordio_2003_permutation}. Here, we adapt the encoder provided in~\cite{nordio_2003_permutation}.
A pseudocode for this algorithm is provided in Algorithm~\ref{alg:enc_general}. A pseudocode for the inverse encoder is given in Algorithm~\ref{alg:inv_enc_general}. The computational complexity of both algorithms is $O(nu)$ arithmetic operations.

\subsection{Variant~II Encoder Function and its Inverse}

\begin{algorithm}
\caption{An inverse for the encoder of Algorithm~\ref{alg:enc_general} that produces the index in the range $[0,\lvert\mathcal{C}\rvert-1]$ of the input permutation in lexicographic order.}\label{alg:inv_enc_general}
\DontPrintSemicolon
\KwIn{index vector $\bm{c}$, vector of repetitions $\bm{m}$, cardinality of the code $\mathit{size}$, code length $n$}
\KwOut{codeword integer index $q$}
$q \gets$ $0$\;
\For{$index \gets 1$ \KwTo $n$}{
  \For{$k \gets 1$ \KwTo $\bm{c}[\mathit{index}]-1$}{
    $\mathit{size}_i \gets \lfloor \mathit{size} \cdot \bm{m}[k]/(n-\mathit{index}+1)\rfloor$\;
    $q \gets q + \mathit{size}_i$\;
  }
  $\mathit{size} \gets \lfloor \mathit{size} \cdot \bm{m}[\bm{c}[\mathit{index}]]/(n-\mathit{index}+1)\rfloor$\;
  $\bm{m}[\bm{c}[\mathit{index}]] \gets \bm{m}[\bm{c}[\mathit{index}]] - 1$\;
}
\Return $q$\;
\end{algorithm}

A Variant~II permutation code $\mathcal{C}$, similar to a Variant~I permutation code, is uniquely identified by its initial vector $\bm{x}$. Note that all elements of the initial vector $\bm{x}$ are positive.
We use the same notation $u$ to denote the number of distinct elements of $\bm{x}$. Similarly, the vector $\bm{\mu}$ is defined as the vector whose elements are the distinct elements of $\bm{x}$ in ascending order. The number of times each $\mu_i$ occurs in $\bm{x}$ is denoted as $m_i$. 
The rate of this code is
\[
\phantom{.}r = 1+\frac{1}{n}\log_2 \frac{n!}{m_1!m_2!m_3!\dots m_u!}.
\]

Let $k = \lfloor nr\rfloor$. Encoding an integer in the range $0$ to $2^k-1$ can be done using an encoder for its constituent Variant~I permutation code. The $k-n$ least significant bits are mapped to a codeword $\bm{c}$ in the constituent Variant~I permutation code using an encoder as in Section~\ref{subsec:varienc}. The signs of the symbols in the resulting codeword $\bm{c}$ are chosen using the $n$ most significant bits. If the encoder for the constituent Variant~I permutation code is $\enc_{\mathrm{I}}$, the encoder for the Variant~II permutation code will be the function 
\[
\enc_{\mathrm{II}}:\{0,1\}^k\to \mathcal{C}
\]
such that
\begin{align*}
&\enc_{\mathrm{II}}(b_1,b_2,b_3,\dots, b_k)=\\&\diag(1\!-\!2b_1,1\!-\!2b_2,\dots,1\!-\!2b_n)\enc_{\mathrm{I}}(q),
\end{align*}
where, base-$2$ representation of $q$ is $b_{n+1}b_{n+2}\dots b_k$.
Similarly, an inverse encoder can be defined using the inverse encoder for the constituent Variant~I permutation code.

\subsection{Permutation Encoding for PAS}\label{subsec:bitmapping}
We are now ready to explicitly describe the first block in the PAS system of Fig.~\ref{fig:pas}. Assume we want to use a Variant~I permutation code with initial vector $\bm{x}$ and blocklength $n$ whose rate is $r$. The first step for us is to expurgate the corresponding Variant~I code so that its size is a power of $2$. Let $k_\mathrm{a} = \lfloor nr\rfloor$. To identify a sequence of amplitudes of length $n$, i.e., a Variant~I codeword of blocklength $n$, one can use the encoder defined in Section~\ref{subsec:varienc}. The input is an integer $i$ in the range $[0, 2^{k_\mathrm{a}})$. 
Therefore, each codeword in this Variant~I permutation code can be indexed by a binary vector of length $k_\mathrm{a}$. This way, if $nr$ is not an integer, the Variant~I permutation code is expurgated so that only the first $2^{k_\mathrm{a}}$ codewords, based on the lexicographical ordering, are utilized. This can potentially introduce a considerable correlation between the elements of codewords and breaks the geometric uniformity of the permutation code~\cite{forney_1991_geometrically}. 

In order to properly randomize so that each codeword is surrounded by approximately equal number of codewords at any given distance, we first try to randomly spread the set of integers in the range of $[0, 2^{k_\mathrm{a}})$ on the integers in the range of $[0, 2^{nr})$. One way to achieve this is to take the sampled integer $i$ from $[0, 2^{k_\mathrm{a}})$ and multiply it by a randomly chosen integer $e$ coprime with $2^{nr}$ and add another randomly chosen integer $d$ in $[0, 2^{k_\mathrm{a}})$ modulo $2^{nr}$. The coprime factor $e$ and the dither $d$ can be generated based on some shared randomness so that the receiver can also reproduce them synchronously. Consequently, the first block in Fig.~\ref{fig:pas} takes in an integer $i$ from $[0, 2^{k_\mathrm{a}})$ at the input and produces a Variant~I codeword $\bm{c}$ by
\[
\phantom{.}\bm{c} = \enc_{\mathrm{I}}(ei+d \pmod {2^{nr}}).
\]
The codeword $\bm{c}$ serves as the set of amplitudes in PAS and is labeled based on the Gray labeling of the constituent PAM constellation. The labels are then encoded using a systematic soft-decodable channel code. The resulting parity bits are used to set the signs of each of the amplitudes and as a result, a Variant~II permutation code is produced.

\begin{example} 
To better understand the significance of randomization, consider the case of using a permutation code $\mathcal{C}$ and choosing the first $\log_2\lfloor\lvert C\rvert\rfloor$ codewords according to the lexicographical ordering of the codewords. For example, assume that the initial vector is
\[
v = (1,1,1,1,1,3,3,3,5,5,5,7).
\]
With this choice, $\lvert \mathcal{C}\rvert \approx 2^{16.76}$. If we only consider the first $2^{16}$ codewords according to the lexicographical ordering, the last codeword that we use is
\[
c = (3, 5, 1, 3, 1, 1, 1, 3, 5, 1, 5, 7).
\]
Note, for instance, that in none of the codewords is the first symbol $5$ or $7$. Also note that the second symbol is never $7$. As is evident, the correlation between symbols is very strong. With the proposed randomization, such strong correlations are avoided to a great extend.
\end{example}

\subsection{Soft-in Soft-out Demapping}\label{sec:decodingMethods}
The input to the first block at the receiver in Fig.~\ref{fig:pas} contains noisy version of Variant~II permutation codewords. Assume that the received word is $\bm{y}$ and that we want to estimate the LLR for the $i^{\mathrm{th}}$ bit, $\mathrm{LLR}_i$, which is a part of the label in the $j^{\mathrm{th}}$ symbol, $c_j$, of the Variant~II codeword $\bm{c}$. Four different estimators are discussed in what follows.

\noindent\textbf{Exact Method:} To get the exact LLR, one should use
\begin{equation}\label{eq:exactllr}
\phantom{,}\mathrm{LLR}_i = \log\frac{\sum_{\bm{c}, b_i = 0}p(\bm{y} \mid \bm{c})}{\sum_{\bm{c}, b_i = 1}p(\bm{y} \mid \bm{c})},
\end{equation}
where, in the numerator, $\bm{c}$ runs over all Variant~II codewords that their $i^{\mathrm{th}}$ bit $b_i$ is $0$, and that 
\[
\phantom{,}e^{-1}(\enc_{\mathrm{I}}^{-1}(\lvert c_1\rvert, \lvert c_2\rvert, \dots, \lvert c_n\rvert) - d) \pmod {2^{nr}} \in [0, 2^{k_\mathrm{a}}).
\]
The condition for codewords $\bm{c}$ in the denominator is similar except $i^{\mathrm{th}}$ bit $b_i$ is $1$. The total number of terms in both summations is $2^{n+k_\mathrm{a}}$ which is at least half of the size of the Variant~II code $\mathcal{C}$ and is typically very large. Note that there are $n\log_2\lvert \mathcal{M}\rvert$ label bits in each codeword and, as a result, the computational complexity of this method is $O(n\lvert \mathcal{C}\rvert\log\lvert \mathcal{M}\rvert)$ operations per permutation codeword. The exact method, therefore, is only useful for very short permutation codes.

\noindent\textbf{Symbol-by-Symbol:} In the symbol-by-symbol method of~\cite{bocherer_2015_bandwidth}, the $i^{\mathrm{th}}$ LLR is estimated using
\begin{equation}\label{eq:sbsllr}
\phantom{,}\mathrm{LLR}_i = \log\frac{\sum_{s, b_i = 0}m_sp(y_j \mid s)}{\sum_{s, b_i = 1}m_sp(y_j \mid s)},
\end{equation}
where $s$ runs over all symbols in the constituent constellation and, with slight abuse of notation, $m_s$ is the number of occurrences of $\lvert s \rvert$ in the initial vector. From the point of view of the receiver, any word from $\mathcal{M}^n$ is a possible codeword and in the computation of the LLRs all words in $\mathcal{M}^n$ can contribute. This allows for a computationally efficient method of estimating LLRs. For each LLR, computation of~(\ref{eq:sbsllr}) requires $O(\lvert \mathcal{M}\rvert)$ operations. Therefore, the computational complexity of the symbol-by-symbol method is $O(n\lvert \mathcal{M}\rvert\log\lvert \mathcal{M}\rvert)$ operations per permutation codeword.
However, this method is inaccurate at short blocklengths~\cite{schulte_2020_joint, luo_2023_joint}. 

\noindent\textbf{Orbit Decoding with Frozen Symbols:} One key contribution of this paper is this new and simple method of soft demapping using orbit decoding of permutation codes. If the $i^{\mathrm{th}}$ bit is an amplitude bit (i.e., it is not a sign bit), to estimate $\mathrm{LLR}_i$ we use
\begin{equation}\label{eq:ourllra}
\phantom{,}\mathrm{LLR}_i = \log\frac{\sum_{\bm{c}, b_i = 0}f_{\mathrm{o}}(\bm{y} \mid \bm{c})}{\sum_{\bm{c}, b_i = 1}f_{\mathrm{o}}(\bm{y} \mid \bm{c})},
\end{equation}
where $\bm{c}$ runs over all Variant~I codewords in the following list: For any amplitude $s\in\{\mu_1,\mu_2,\dots,\mu_u\}$, we only include the codeword $\bm{c}$ with the most likely orbit such that $c_j=s$ on the list. That is, the most likely orbits for every element of the codeword frozen to any possible amplitude will contribute in the calculation of $\mathrm{LLR}_i$. 

If the $i^{\mathrm{th}}$ bit is a sign bit, we only take into account the contribution of the half of the codewords in each orbit that have positive sign in the numerator and the other half in the denominator. Using a technique similar to the proof of Theorem~\ref{thm:orbitlikelihood}, one can show that for a sign bit
\begin{equation}\label{eq:ourllrs}
\phantom{.}\mathrm{LLR}_i = \log\frac{\sum_{\bm{c}}e^{\frac{c_jy_j}{\sigma^2}}\sech\frac{c_jy_j}{\sigma^2} f_{\mathrm{o}}(\bm{y} \mid \bm{c})}{\sum_{\bm{c}}e^{-\frac{c_jy_j}{\sigma^2}}\sech\frac{c_jy_j}{\sigma^2}f_{\mathrm{o}}(\bm{y} \mid \bm{c})}.
\end{equation}

To find the corresponding orbit of each codeword on the list, the assignment problem needs to be solved with the frozen symbol as an inclusion constraint (See the Appendix). Once the received word is sorted by the magnitude of its elements, the orbit likelihood of each item on the list can be found using $O(n)$ operations (see Theorem~\ref{thm:orbitlikelihood}). There are a total of $O(n\lvert \mathcal{M}\rvert)$ orbits to consider. Therefore, the orbit likelihoods can be computed with $O(n^2\lvert \mathcal{M}\rvert)$ complexity. Once the orbit likelihoods are computed, each of the $n\log_2\lvert \mathcal{M}\rvert$ LLRs can be found with $O(\lvert\mathcal{M}\rvert)$ complexity. 
Therefore, the computational complexity of this method is $O(n^2\lvert \mathcal{M}\rvert + n\lvert \mathcal{M}\rvert \log\lvert \mathcal{M}\rvert)$ operations per permutation codeword.
\begin{example}
Consider a permutation code with the initial vector
\[
\phantom{.}(1,1,1,1,1,1,1,3).
\]
The corresponding Variant~I code has size $2^3$, and since $nr = 3$ is an integer, there is no need for random spreading. Assume that we send the codeword 
\[
\bm{x} = (1,1,1,1,1,1,1,3)
\]
and the noise vector, sampled from a standard normal random distribution and rounded to one significant digit, is
\[
\phantom{.}\bm{z} = (2.1, 0.2, 0.1, 1.5, 0.7, 1.6, -1.9, 0.2).
\]
Assume that the constituent constellation is labeled such that the amplitude of $1$ is labeled with bit $0$ and the amplitude of $3$ is labeled with bit $1$. The LLR for the first amplitude bit using the exact method is $0.69$. The estimate for this LLR using the symbol-by-symbol method is $-0.25$. The estimate for the LLR using orbit decoding is $0.20$. Note that the sign of the LLR using the symbol-by-symbol method is more in favor of an amplitude of $3$, while the other two values are in favor of an amplitude of $1$ in the first position. This shows that in this example the quality of the LLR estimated using the orbit decoding with frozen symbols is superior to the estimation using the symbol-by-symbol method.
\end{example}

\noindent\textbf{Soft Decoding Using the BCJR Algorithm:} For typical permutation code parameters of interest such as Code 2 of Example~\ref{ex:mi12}, the natural trellis representation of permutation codes has a prohibitively large number of states and edges~\cite{kschischang_1996_thetrellis, schulte_2020_joint}. However, the trellis representation of the more general spherical codes that we study in Section~\ref{sec:shell}, which include permutation codes as subcodes, is much more tractable. Similar to the symbol-by-symbol method, in which the receiver presumes a codebook encompassing the used permutation code, we can consider the larger, more general spherical codes and estimate the likelihood of each symbol in the codeword by running the Bahl--Cocke--Jelinek--Raviv (BCJR) algorithm~\cite{bahl_1974_optimal}. 

\subsection{Numerical Evaluation}

Consider the abstract channel for which the channel input is the label of the transmitted permutation codeword $\bm{B}=(B_1,B_2,\dots,B_{k_{\mathrm{a}}+n})$ and its output is the noisy permutation codeword $\bm{Y}=(Y_1,Y_2,\dots,Y_n)$. The bit-metric decoding (BMD) rate\footnote{Here, $\left[x\right]^+$ is the same as $\max \{x,0\}$.}~\cite{bocherer_2023_probabilistic}
\[
\left[H(\bm{B}) - \sum_i H(B_i \mid Y)\right]^+
\]
is an achievable rate~\cite{bocherer_2014_achievable} and can be used to compare different methods of soft demapping.
This is particularly useful as it provides an achievable rate when a binary soft-decodable channel code is intended to be used, without specifically restricting the system to a particular implementation of the channel code.
The BMD rate of permutations codes, as well as generalizations discussed in the next section, for different methods of estimating LLRs are shown in Fig.~\ref{fig:bmd}.

\begin{figure}
    \centering
    \includegraphics[width=\columnwidth]{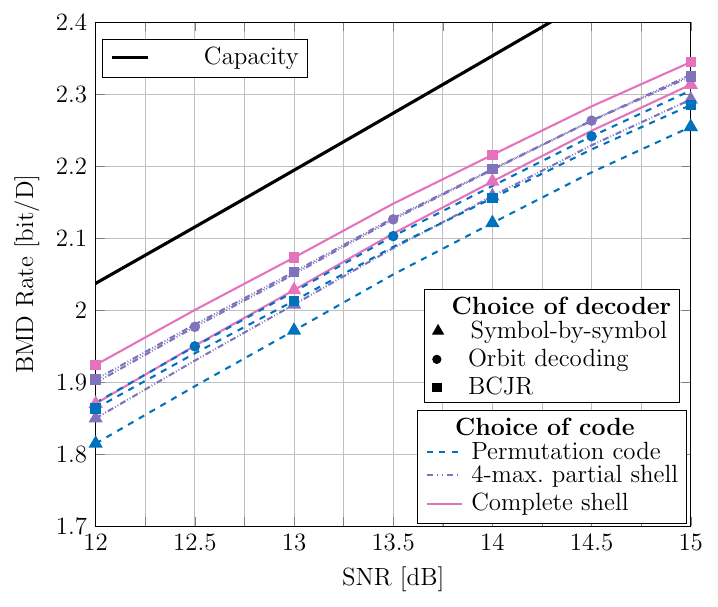}
    \caption{BMD rates for various spherical codes with $n=50$ and $\mathtt{E}=530$ with a $8$-PAM constituent constellation are shown. }
    \label{fig:bmd}
\end{figure}

To better compare the symbol-by-symbol method with our proposed method using orbit decoding as well as decoding using the BCJR algorithm, we use a permutation code in a PAS scheme. We compare the permutation block error rate (BLER) for Code 2 from Example~\ref{ex:mi12} for these three estimation methods.
By each block, we mean a permutation codeword (or, in general, a spherical codeword) and BLER is the ratio of
blocks received in error after the channel decoder to the total number of blocks sent.
For the channel code, we use the $(10\,860, 8\,448)$ low density parity check (LDPC) code from the 5G~NR standard~\cite{5g_2018_ldpc}. Each LDPC code frame contains $67$ permutation codewords. Among the sign bits, $14$ are used for systematic information bits~\cite{bocherer_2015_bandwidth}. In Fig.~\ref{fig:BLER}, we show the BLER for these three LLR estimation methods. Note that we think of the case of using a permutation code with symbol-by-symbol decoding~\cite{bocherer_2015_bandwidth} as the baseline for comparison with other spherical codes and decoding methods that we study. For instance, it is clear that at BLER~$=10^{-3}$, about $0.3$~dB improvement is achieved by using the orbit decoding with frozen symbols compared with the symbol-by-symbol method. Also, about $0.2$~dB improvement is achieved compared with the BCJR method. In Fig.~\ref{fig:BLER}, similar results are shown for more general spherical codes as well. The structure of these codes will be explained in Section~\ref{sec:shell}.

\section{Structure of Constant-Energy Shells of PAM Constellations}
\label{sec:shell}
In this section, we explore a generalization of permutation codes that consists of codewords from a union of permutation codes, each corresponding to a distinct initial vector or ``type".  If each of these initial vectors has the same energy, the resulting code is termed a ``spherical code".  In comparison to permutation codes, this construction can deliver an increase in code rate.  At the same time the codes we now describe are also spherical codes and so, as mentioned in the introduction, can reduce the XPM in the optical fiber channel.  We are particularly interested in constructions of spherical codes that abide by the constraints enforced by standard (PAM) modulation formats.  Given an underlying constellation and a desired codeword energy (squared sphere radius), the largest code can be designed is the intersection of a sphere of the given radius with the $n$-fold product of the underlying constellation.

Larger spherical codes, while often superior in rate, can have high complexity for storage, encoding, and decoding.  For this reason, for a given PAM constellation, we also investigate the structure of these codes so as to describe intermediate subcodes that can increase code performance without drastically increasing complexity.

\begin{figure}
    \centering
    \includegraphics[width=\columnwidth]{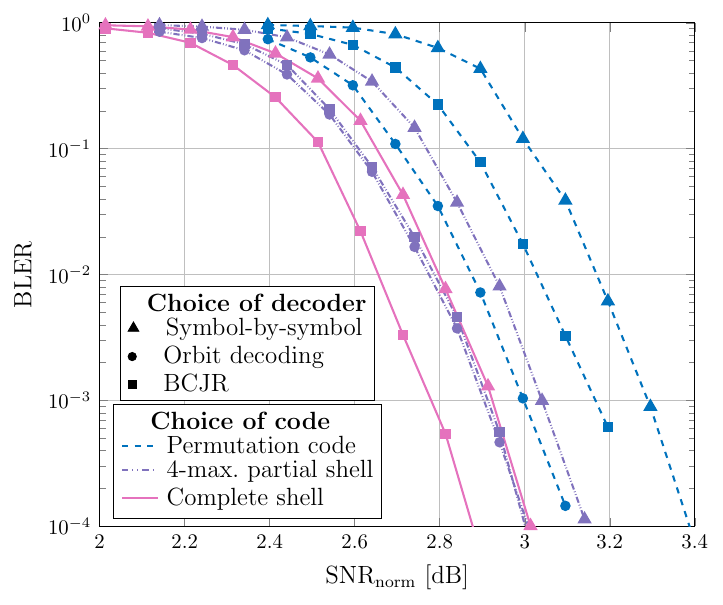}
    \caption{Block error rates for various spherical codes with $n=50$ and $\mathtt{E}=530$ with a $8$-PAM constituent constellation. The $(10\,860, 8\,448)$ LDPC code from the 5G~NR standard is used in the PAS architecture.}
    \label{fig:BLER}
\end{figure}

\subsection{Complete Shell Codes}

We design codes with blocklength $n$ that are shells of the product of $n$ $2p$-PAM constellations, constrained to have total energy $\mathtt{E}$ in each block.  This is formalized as follows.

\begin{definition}
Given a symbol set $\mathcal{M}$, a blocklength-$n$ code consisting of all codewords of the form $(x_1, x_2, \dots, x_n)$, where the $x_i$s satisfy $x_i \in \mathcal{M}$ and 
$\sum_{i = 1}^n x_i^2 = \mathtt{E}$ is referred to as a \emph{complete shell code} over $\mathcal{M}$.
\end{definition}

We refer to the complete shell code of length $n$ and squared radius $\mathtt{E}$ over the $2p$-PAM constellation as the $(n, \mathtt{E}, p)$ code. Complete shell code automatically deliver a shaping gain, as the symmetry of the construction induces an approximate Boltzmann distribution over the symbols. That is, if codewords are selected uniformly, the probability distribution over the symbols $x_i$ is $p(x_i) \propto \exp[\lambda x_i^2]$, where $\lambda$ is a function of the code parameters $n$ and $\mathtt{E}$. Note that this is a discretized version of a Gaussian distribution.

Complete shell codes can be decomposed into a union of permutation subcodes as $\mathcal{C} = \cup_{i = 1}^{t(\mathcal{C})}\mathcal{C}^i$, where each $\mathcal{C}^i$ is a permutation code, and the number of distinct permutation subcodes $t(\mathcal{C})$ is a function of the code parameters.  We will also write $t(\mathcal{C}) = t(n, \mathtt{E}, p)$ when $\mathcal{C}$ is the $(n, \mathtt{E}, p)$ code.  We refer to each subcode as a \emph{type class} following the terminology of, e.g.,~\cite[Ch.~11.1]{cover_1999_elements}.  Describing the code in this union form is useful when describing and analyzing the code performance.

We call the initial vector in a type class the \textit{type class representative}.  For instance, the $(8, 32, 4)$ code has two type classes, with representatives $(1, 1, 1, 1, 1, 3, 3, 3)$ and $(1, 1, 1, 1, 1, 1, 1, 5)$.  The first class contains $14\,336 \approx 2^{13.8}$ codewords, and the second $2048 = 2^{11}$.  In general, the number of type classes $t(n, \mathtt{E}, p)$ of a given energy $\mathtt{E}$ is a complicated function, which we usually compute numerically.  A loose bound for the maximum value of $t(n, \mathtt{E}, p)$ when varying $n$ and $p$ is $O(n^p)$.  This becomes intractable quickly for larger values of $p$.

\subsection{Encoding and Decoding Complete Codes}

\subsubsection{Encoding}

It is natural to extend the encoding and decoding methods described in Section~\ref{sec:softdecode} to the complete shell codes. For encoding, if the $t$ type-class subcodes are $\mathcal{C}^1, \mathcal{C}^2, \dots, \mathcal{C}^t$, we can assign the first $\lvert\mathcal{C}^1\rvert$ codeword indices to the first type class, the next $\lvert\mathcal{C}^2\rvert$ to the second, and so forth. This is formalized in Algorithm~\ref{alg:enc_types}. The index of the corresponding permutation subcode can be found using binary search on the vector cumulative sums $\mathit{CS} = (\lvert\mathcal{C}^1\rvert, \lvert\mathcal{C}^1\rvert+\lvert\mathcal{C}^2\rvert, \dots, \sum_{i=1}^T\lvert\mathcal{C}^i\rvert)$.
Since the complexity of Algorithm~\ref{alg:enc_general} is $O(np)$ (as $u = p$ for the PAM case), our new algorithm has complexity $O(np +  t(n, \mathtt{E}, p)) \approx O(n^p)$. 
The randomization method of Section~\ref{subsec:bitmapping} can be used here as well. The index input to this algorithm is obtained similarly to Section~\ref{subsec:bitmapping}, with the exception that the code rate $r$ is the rate of the complete shell code.

\begin{algorithm}
\caption{A generalization of Algorithm~\ref{alg:enc_general} to the case of the shell code. The output is a vector of indices $\bm{c}$ of length~$n$ and a type class index $k$ between 1 and $t$. That is, the corresponding codeword is $\bm{\mu}^k[\bm{c}]$.}\label{alg:enc_types}
\DontPrintSemicolon
\KwIn{codeword integer index $q$, vectors of repetitions $\bm{m}^1, \bm{m}^2, \dots, \bm{m}^t$, cumulative cardinalities of the subcodes $\mathit{CS}$, code length $n$}
\KwOut{index vector $\bm{c}$, type class index $k$}
Find the index $k$ such that $\mathit{CS}[k-1]<q\leq \mathit{CS}[k]$.
$\mathit{size} \gets \mathit{CS}[k] - \mathit{CS}[k-1]$\;
$c \gets enc_I(q - \mathit{CS}[k-1]; \bm{m}^k, \mathit{size}, n)$\;
\Return{$c, k$}\;
\end{algorithm}

\subsubsection{Orbit decoding}

For decoding, we again consider our code as a union of permutation codes (type classes).  When forming the list of most likely orbits for each frozen symbol, we find the most likely orbit for each type class and pick the one with the highest likelihood to put on our list. Once the received word is sorted by the magnitude of its elements, the orbit likelihood of each item on the list can be found using $O(n)$ operations. There are a total of $O(npt(n, \mathtt{E}, p))$ orbits to consider. Once the list of orbits with the highest likelihoods is formed, each of the $O(n\log_2 p)$ LLRs can be computed with $O(p)$ complexity. Therefore, the total computational complexity per spherical codeword is $O(n^2pt(n, \mathtt{E}, p) + np\log p)\approx O(pn^p)$.

\subsubsection{BCJR Decoding}

A natural alternative decoding procedure to consider in this context is BCJR decoding, as described in~\cite{bahl_1974_optimal}.  The complete shell code has a simple trellis structure, similar to that discussed in~\cite{gutelkin_2020_enumerative}.  An example trellis for the $(8, 32, 3)$ code is shown in Fig.~\ref{fig:trellis}.  Here, each path from left to right through the trellis represents an unsigned codeword, with states labeled by the energy accumulated up to that point in the code.  It can be shown that at depths between 3 and $n - 3$, the trellis will have at most $(\mathtt{E} - n)/8$ states.  Since each state has a maximum of $p$ outgoing edges, the number of edges in the trellis is $O(n(\mathtt{E} - n)p)$, which is the complexity of the algorithm.

\begin{figure}
    \centering
    \includegraphics[width=\columnwidth]{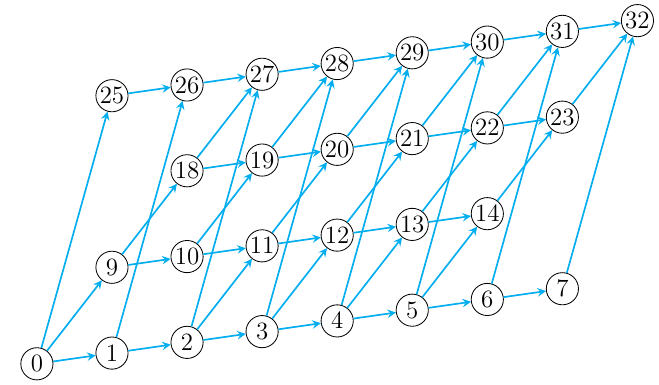}
    \caption{The trellis for the $(8, 32, 3)$ code.}
    \label{fig:trellis}
\end{figure}

\subsubsection{Results}

If we accept our $O(n^{p - 2})$ heuristic for $t(n, \mathtt{E}, p)$, we see that orbit decoding has a much worse complexity than BCJR decoding, meaning that BCJR decoding is preferable.  In Fig.~\ref{fig:BLER} we see the results of decoding the complete $(50, 530, 4)$ shell code, using both symbol-by-symbol decoding and BCJR decoding.  
Results are presented in terms of the rate-normalized SNR which is defined as follows.  
The channel coding theorem for the real AWGN channel indicates that a coding scheme with rate $\rho$ can achieve an arbitrarily small probability of error if and only if 
\[
\rho \leq \frac{1}{2}\log_2(1+\mathrm{SNR}).
\]
Equivalently, a coding scheme with rate $\rho$ can achieve an arbitrarily small probability of error if and only if
\[
\frac{\mathrm{SNR}}{2^{2\rho}-1}\geq 1.
\]
The quantity $\frac{\mathrm{SNR}}{2^{2\rho}-1}$ is defined as the rate-normalized SNR $\mathrm{SNR}_{\mathrm{norm}}$~\cite{forney_1998_modulation}.

Symbol-by-symbol decoding is not optimal.  Furthermore, the complete code outperforms the permutation code.

\subsection{Partial Shell Spherical Codes}

For large $n$, orbit decoding is impractical. Even BCJR decoding becomes unwieldy as $\mathtt{E}$ grows.  Here we define a new family of codes that has competitive rate and shaping properties, with vastly improved encoding and decoding complexities.  In particular, keeping only a small number of the largest type classes in a complete shell code is often sufficient to create a high-performing code.  We formalize this as follows.

\begin{definition}
Take a complete shell code $\cup_{i = 1}^{t(\mathcal{C})}\mathcal{C}^i$, where each $\mathcal{C}^i$ is a distinct type class with symbols drawn from a modulation format $\mathcal{M}$, and $\lvert\mathcal{C}^1\rvert \ge \lvert\mathcal{C}^2\rvert \ge \dots \ge \lvert\mathcal{C}^t\rvert$. The subcode of $\mathcal{C}$ defined by $\mathcal{C}^1\cup \mathcal{C}^2 \cup \dots \cup \mathcal{C}^k$ is referred to as the \emph{maximal $k$-class partial shell code} over $\mathcal{M}$.
\end{definition}

The maximal $k$-class partial shell code consists of the $k$ largest type classes.  In many cases, relatively small values of $k$ result in codes whose rates are only marginally smaller than that of the complete shell code, with the benefit that the encoding and decoding algorithms now have a constant number of type classes, i.e., we have replaced the function $t(\mathcal{C})$ with the constant $k$.  Note that $k = 1$ corresponds to a permutation code.

\begin{example}\label{ex:selective}
Consider the $(50, 530, 4)$, $(25, 305, 4)$ and $(100, 996, 4)$ codes.  While these codes have different block-lengths, $\mathtt{E}$ was chosen to give similar numbers of bits per dimension.  In each case, we consider the $k$ subcodes for different values of $k$, comparing the rate of the partial shell code to the rate of the complete code.  In Fig.~\ref{fig:best_classes} the fractional rate loss is very small even when $k$ is small.  These codes have 34, 113, and 369 type classes respectively, but values of $k$ below 10 already achieve 99\% of the rate.
\end{example}

\begin{figure}
    \centering
    \includegraphics[width=\columnwidth]{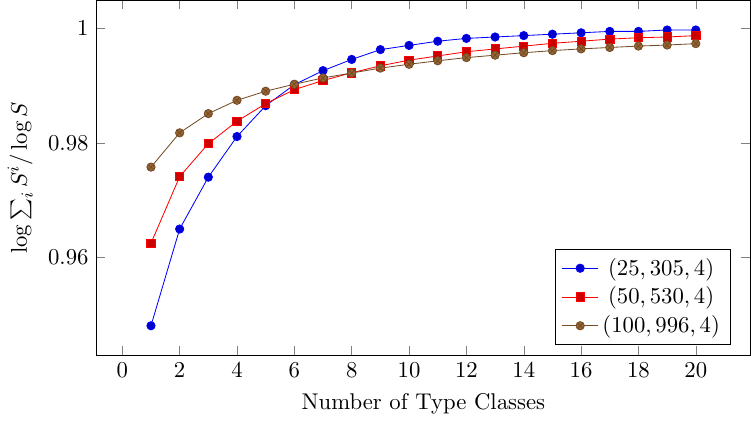}
    \caption{A comparison of sizes of subcodes constructed from the largest type classes in the indicated codes.}
    \label{fig:best_classes}
\end{figure}

The acceptable reduction in rate will depend on the use-case for the code, but we already see the complexity benefits of using selective codes instead of complete codes.  Conversely, if we want to create a selective code from a union of permutation codes, we see that the rate increases for adding each subsequent code to the union decrease quickly.

As previously noted, one of the appealing features of complete codes over PAM constellations is that they induce an approximate marginal Boltzmann distribution.  Partial shell codes also have this property.  This is a result of the fact that they are constructed from the largest type classes available, whose symbols $m_i$ are distributed proportionally to $\exp[\lambda(2i - 1)^2]n$, a Boltzmann distribution, where parameter $\lambda$ depends on $\mathtt{E}$.  This result (more specifically its dual) is shown in~\cite{ingemarsson_1990_optimized}.

The exact Boltzmann distribution does not give an integer solution in general, but the non-integer solution is often a near approximation for solutions.  In fact, since we are often interested not only in the largest type class, but in the $k$ largest type classes, an alternative method for constructing $k$-class partial shell codes is to start with the non-integer Boltzmann solution and search for several nearby solutions, e.g., by iterating over all combinations of integer values within some neighborhood of each non-integer entry, and keeping the resulting vectors that are the $k$ best solutions.  

\begin{example}
As an example, for the $(n, \mathtt{E}, p) = (50, 530, 4)$ code, whose largest type class is Code 2 in Example~\ref{ex:mi12}, we have the following exact Boltzmann distribution:
\begin{equation}
(m_1, m_2, m_3, m_4) = (22.38, 16.12, 8.37, 3.13)
\end{equation} with $\lambda = -0.04$.  The largest type classes in the integer-constrained case are given in Table~\ref{tab:classEx}.  The distributions are similar to the non-integer solution, and can be found by searching for integer solutions in its neighborhood.  Each type class in this table contains approximately $2^{79}$ codewords, not counting sign flips (i.e., in the Variant~I version).  Most of the balance of the 113 type classes in the complete shell code are much smaller and contribute little, as seen in Example~\ref{ex:selective}.
\end{example}

\begin{table}
\begin{center}
\begin{tabular}{|c|c|}
\hline
$(m_1, m_2, m_3, m_4)$ & Bits\\
\hline
(24, 15, 7, 4) & 79.87\\
(21, 18, 8, 3) & 79.40\\
(23, 15, 9, 3) & 78.45\\
\hline
\end{tabular}
\end{center}
\caption{The largest type classes for the $(50, 530, 4)$ code.}
\label{tab:classEx}
\end{table}

\subsection{Encoding and Decoding Partial Codes}

The primary benefit of partial codes over complete codes is their improved complexity of decoding. Encoding can be done using Algorithm~3 with using $k$ type-classes.% encoding and decoding.

\subsubsection{Encoding}
The encoder is the same as the encoder in Algorithm~\ref{alg:enc_types}, with the one change that $t(n, \mathtt{E}, p)$ is replaced with constant $k$.  This now has complexity $O(np)$, which is much more practical. The randomization method of Section~\ref{subsec:bitmapping} can be used here as well.

\subsubsection{Type-class decoding}

The orbit decoding algorithm is also the same as the algorithm for the complete case.  Again, $t(n, \mathtt{E}, p)$ is replaced with $k$, which gives a complexity of $O(np(n + \log p))$.

\subsubsection{BCJR decoding}

Finally, we can still use BCJR decoding in the partial shell case.  The trellis for the partial code is complicated, but we can decode more efficiently by using the trellis for the complete code, and reporting an error if the output is not a codeword.  This still has complexity $O(np(\mathtt{E} - n))$.

Comparing complexities, we see that orbit decoding is now preferable when $n + \log p \in o(p(\mathtt{E} - n))$, which happens as $E$ becomes large for fixed $n$ and $p$.

\subsection{Results}
Fig.~\ref{fig:BLER} also shows the SNR values for the 4-maximal partial shell code. We see that they are intermediate between the results for the permutation and complete codes. It is important to note that, compared with orbit decoding, the complete shell code with symbol-by-symbol decoding performs very well and is less than 0.15~dB worse than decoding with the BCJR algorithm. In terms of computational complexity, this is much cheaper than the BCJR algorithm and, therefore, the complete shell code with symbol-by-symbol decoding is still practically competitive.

\section{Conclusions}
In this paper, we have studied the use of more general spherical codes in place of permutation codes for probabilistic amplitude shaping at short blocklengths, which are of particular relevance to optical communications. A key result is the development of a novel soft-demapping technique for these short spherical codes. When used in a PAS framework with a standard LDPC channel code, our findings suggest that a gain of about $0.5$~dB can be achieved compared to the conventional case of using one permutation code with the symbol-by-symbol decoding. In terms of BMD rate, the total gain is calculated to be more than $0.6$~dB for a particular spherical codes at blocklength $n=50$.

However, while our results are promising, it remains to be seen how well our proposed methods work for nonlinear optical fiber. Future work will focus on further exploring this aspect, while studying possible improvements and general performance versus complexity trade-offs for our proposed methods.
%%%%%%%%%%%%%%%%%%%%%%%%%%%%%%%%%%%%%%%%%%%
\begin{figure}
\centering
\includegraphics[width=0.75\columnwidth]{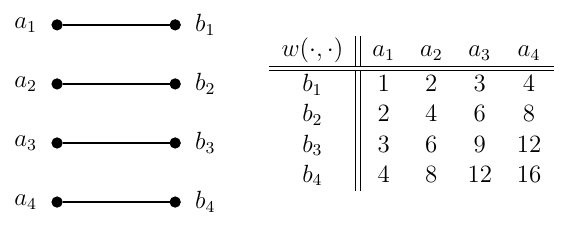}
\caption{The matching corresponding to the solution with the largest reward is shown on the left. The weight function for the assignment problem of our running example is also shown on the right.}
\label{fig:example1}
\end{figure}
%%%%%%%%%%%%%%%%%%%%%%%%%%%%%%%%%%%%%%%%%%%

\appendix[Computation of Mutual Information] 
One can calculate the mutual information between input and output of the channel by Monte Carlo evaluation of the mutual information integral given by
\[
\phantom{,}\int \sum_{\bm{x}\in \mathcal{C}}p(\bm{x},\bm{y})\log\frac{p(\bm{y}\mid \bm{x})}{p(\bm{y})}\,\mathrm{d}\bm{y}.
\]
To find $p(\bm{y})$ for each trial, one may use the total law of probability
\[
\phantom{.}p(\bm{y}) = \sum_{\bm{x}\in \mathcal{C}}p(\bm{y}\mid \bm{x})p(\bm{x}).
\]
The code $\mathcal{C}$ is typically very large, meaning that this summation cannot be computed exactly. Instead, we use a list-decoding approximation as follows:
\begin{equation}\label{eq:approxpy}
p(\bm{y})\approx\sum_{\bm{x}\in \mathcal{L}}p(\bm{y}\mid \bm{x})p(\bm{x})
\end{equation}
where $\mathcal{L} = \{\bm{c}_0, \bm{c}_1, \dots, \bm{c}_{L-1}\}$ is a list of $L$ codewords with the highest likelihoods so that
\[
\phantom{.}p(\bm{y}\mid \bm{c}_0) \geq p(\bm{y}\mid \bm{c}_1) \geq \dots \geq p(\bm{y}\mid \bm{c}_{L-1}). 
\]
If $L$ is large enough, $p(\bm{y})$ can be approximated accurately. 
Because of the exponential decay of a Gaussian distribution, usually a computationally feasible list size is enough to get an accurate estimate of $p(\bm{y})$. Note that higher SNRs require smaller list sizes. 

Murty's algorithm finds $L$ solutions to the assignment problem with the largest rewards~\cite{murty_1968_algorithm}. We describe a specialized version of Murty's algorithm pertinent to the list decoding problem of this paper using a running example. The first step is to find the maximizing solution $\bm{c}_0$. 
\begin{example}
We consider the assignment problem with 
\begin{align*}
\mathcal{U} &= \{a_1, a_2, a_3, a_4\},\\
\mathcal{V} &= \{b_1, b_2, b_3, b_4\},
\end{align*}
and the weight function $w(a_i,b_j) = ij$ as our running example, as shown in Fig.~\ref{fig:example1}. We want to find 
\[\phantom{,}\{\bm{c}_0,\bm{c}_1,\bm{c}_2,\bm{c}_3,\bm{c}_4\},\]
the $5$ solutions to the assignment problem with largest rewards.  The matching representing the solution $\bm{c}_0$ is also shown in Fig.~\ref{fig:example1}.
\end{example}

%%%%%%%%%%%%%%%%%%%%%%%%%%%%%%%%%%%%%%%%%%%
\begin{figure}
\centering
\includegraphics[width=0.5\columnwidth]{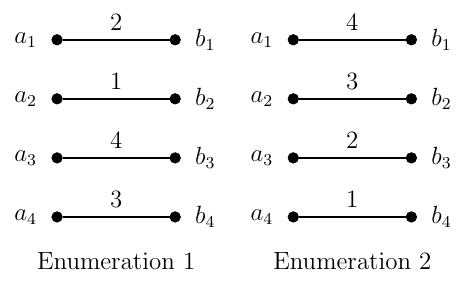}
\caption{Two different enumerations for the edges in $\bm{c}_0$ are shown.}
\label{fig:example_index}
\end{figure}
%%%%%%%%%%%%%%%%%%%%%%%%%%%%%%%%%%%%%%%%%%%

The next step in Murty's algorithm is to choose an arbitrary enumeration\footnote{An enumeration of a non-empty finite set $A$ is a bijection $f_A:[\lvert A\rvert]\to A$. In other words, an enumeration is an ordered indexing of all the items in the set $A$ with numbers $1,2,\dots,\lvert A\rvert$.} of the edges in the matching $\bm{c}_0$.  Then we partition the set of remaining matchings into $n-1$ subsets which we refer to as \emph{cells}. Each cell is specified by a set of \emph{inclusion} and \emph{exclusion} constraints. The first cell contains all matchings that exclude edge $1$ in $\bm{c}_0$. The second cell contains all matchings that include edge $1$ but exclude edge $2$. The next cell contains all matchings that include both edges $1$ and $2$ but exclude edge $3$. The algorithm goes on to partition the set of matchings similarly according to the enumeration of the edges in $\bm{c}_0$. A weight function is defined for each cell by modifying the original weight function based on the inclusion and exclusion constraints. For an excluded edge, the corresponding weight is set to $-\infty$. For an included edge, the two vertices incident with that edge are removed from the vertex sets $\mathcal{U}$ and $\mathcal{V}$ and the corresponding row and column of the weight function are removed. The assignment problem is then solved for each cell with respect to the modified weight function. 

\begin{example}
We consider two different enumerations for the edges in $\bm{c}_0$. These are shown in Fig.~\ref{fig:example_index}. The cells obtained by partitioning according to both enumerations are shown in Fig.~\ref{fig:example_ind1} and Fig.~\ref{fig:example_ind2}, respectively. For each cell, the inclusion constrains are shown with blue dotted lines and the exclusion constraints are shown with red dashed lines.
\end{example}

%%%%%%%%%%%%%%%%%%%%%%%%%%%%%%%%%%%%%%%%%%%
\begin{figure}
\centering
\includegraphics[width=\columnwidth]{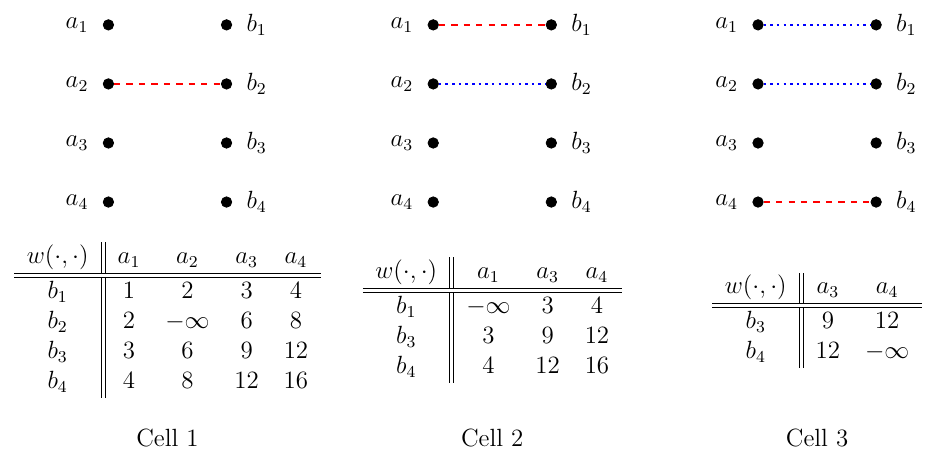}
\caption{The cells obtained after partitioning with respect to Enumeration $1$ for $\bm{c}_0$.}
\label{fig:example_ind1}
\end{figure}
%%%%%%%%%%%%%%%%%%%%%%%%%%%%%%%%%%%%%%%%%%%

The algorithm then proceeds by solving the assignment problem for each cell. Due to the introduction of $-\infty$ into the range of weight function, it should be clear that the weight function of the resulting assignment problems are not necessarily multiplicative. Therefore, solving the assignment problem for each cell may not be as easy as sorting. However, if the enumeration of the edges in $\bm{c}_0$ is chosen carefully, solving the resulting assignment problem may still be easy. In particular, if we can start off by matching the vertices in $\mathcal{U}$ that correspond to columns in the weight function with $-\infty$ entries, we can reduce the problem again to an assignment problem with a multiplicative weight function. To formalize this intuition, the following definition will be useful. 

\begin{definition}
Let $\mathcal{H}$ be a perfect matching of the balanced bipartite graph $G=(\mathcal{U},\mathcal{V},\mathcal{E})$. Let $w_{\mathcal{U}}:\mathcal{U}\to\mathbb{R}\cup\{-\infty\}$ be a weight function for the vertices in $\mathcal{U}$. An enumeration $f_{\mathcal{H}}$ of $H$ is said to be ordered with respect to $w_{\mathcal{H}}$ if
\[
\phantom{.}f_{\mathcal{H}}(a_1,b_1) > f_{\mathcal{H}}(a_2,b_2) \iff w_{\mathcal{U}}(a_1) \leq w_{\mathcal{U}}(a_2).
\]
\end{definition}
In other words, an enumeration is ordered with respect to the weight function of $\mathcal{U}$, if it sorts the edges based on the weight of the vertices in $\mathcal{U}$ that are incident with the edges in descending order. 
\begin{definition}
A weight function $w$ is \emph{almost multiplicative} if there exist two functions
\begin{align*}
w_{\mathcal{U}}:\mathcal{U}\to\mathbb{R}\cup\{-\infty\},\\
w_{\mathcal{V}}:\mathcal{V}\to\mathbb{R}\cup\{-\infty\},
\end{align*}
such that $w_{\mathcal{U}}$ assumes its maximum at $a^*\in \mathcal{U}$,
\[
\phantom{,}\forall b\in \mathcal{V} : w(a^*,b) \in\{-\infty, w_{\mathcal{U}}(a^*)w_{\mathcal{V}}(b)\},
\]
and
\[
\phantom{.}\forall (a,b) \in \mathcal{U}\times \mathcal{V}\setminus\{(a^*,b)\mid b\in \mathcal{V}\} : w(a,b)=w_{\mathcal{U}}(a)w_{\mathcal{V}}(b).
\]
\end{definition}
%%%%%%%%%%%%%%%%%%%%%%%%%%%%%%%%%%%%%%%%%%%
\begin{figure}
\centering
\includegraphics[width=\columnwidth]{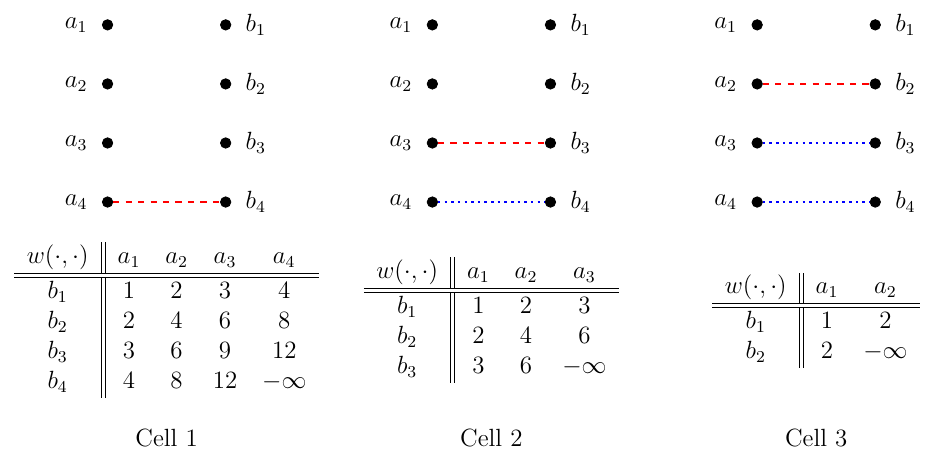}
\caption{The cells obtained after partitioning with respect to Enumeration $2$ for $\bm{c}_0$.}
\label{fig:example_ind2}
\end{figure}
%%%%%%%%%%%%%%%%%%%%%%%%%%%%%%%%%%%%%%%%%%%
If the weight function of an assignment problem is multiplicative and an ordered enumeration with respect to the weight function of $\mathcal{U}$ is used, the weight function of each cell after the first partitioning in Murty's algorithm is almost multiplicative. Solving the assignment problem with an almost multiplicative weight function is straightforward. One can show that in a maximum weight matching, the edge incident on $a^*$ with the largest weight is included in the matching. After removing the corresponding row and column from the weight function, the problem then reduces to an assignment problem with a multiplicative weight function and can be solved by sorting.

%%%%%%%%%%%%%%%%%%%%%%%%%%%%%%%%%%%%%%%%%%%
\begin{figure}
\centering
\includegraphics[width=\columnwidth]{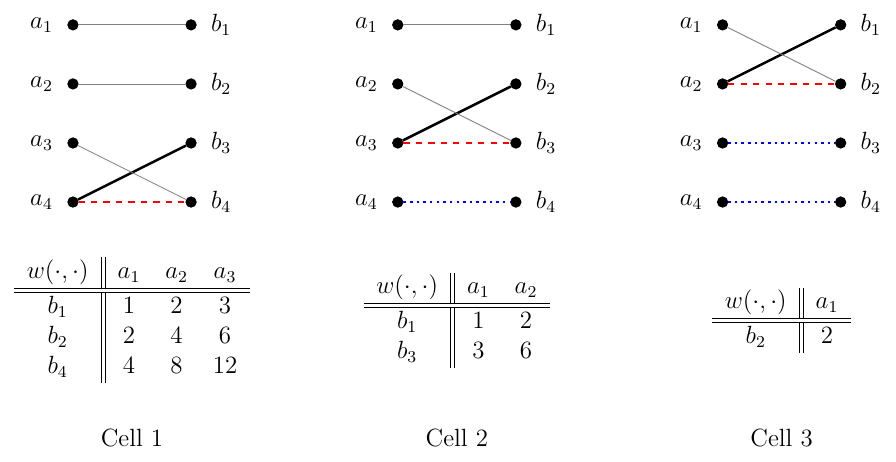}
\caption{The solution of the assignment problem for each cell is obtained after partitioning with respect to Enumeration $2$.}
\label{fig:example_step1}
\end{figure}
%%%%%%%%%%%%%%%%%%%%%%%%%%%%%%%%%%%%%%%%%%%

\begin{example}
Enumeration $2$ in our running example is ordered with respect to $w_{\mathcal{U}}$. In what follows, we only use this ordered enumeration. In each of the cells, the edge connected to vertex in $\mathcal{U}$ with the largest vertex weight is first selected. This is shown by the thick black edge in Fig.~\ref{fig:example_step1}. The weight function is correspondingly adjusted by removing one row and one column. After choosing the first edge, the rest of the edges in the maximum weight matching are found by sorting and are shown using thin gray lines.
\end{example}

A sorted list of potential candidates $\bm{p}$, sorted by their reward, is formed and all the solutions obtained for each cell are placed on this list. Recall that the solution with the largest reward is denoted $\bm{c}_0$. The solution with the second largest reward, $\bm{c}_1$, is the first item in $\bm{p}$, i.e., the potential candidate in $\bm{p}$ with the largest reward. The algorithm removes this solution $\bm{c}_1$ from $\bm{p}$ and partitions the set of remaining matchings in its cell according to an enumeration of $c_1$. Again, by choosing an ordered enumeration with respect to $w_{\mathcal{U}}$, the reduced problems can be solved by sorting.
%%%%%%%%%%%%%%%%%%%%%%%%%%%%%%%%%%%%%%%%%%%
\begin{figure}
\centering
\includegraphics[width=\columnwidth]{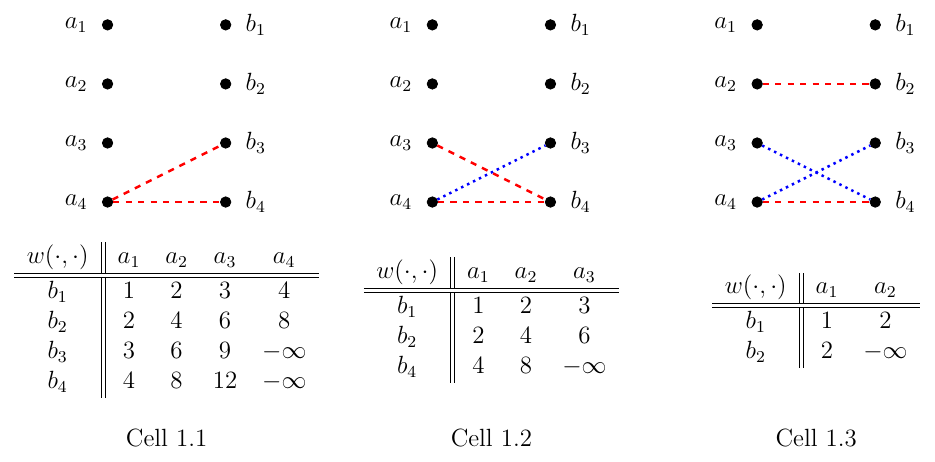}
\caption{The cells obtained after partitioning Cell~$1$ with respect to $\bm{c}_1$.}
\label{fig:example_step21}
\end{figure}
%%%%%%%%%%%%%%%%%%%%%%%%%%%%%%%%%%%%%%%%%%%
\begin{example}
We insert the solutions for each cell into the sorted list $P$. The reward of the solutions of Cell~$1$, Cell~$2$ and Cell~$3$ in Fig.~\ref{fig:example_step1} are all the same and equal to $29$. Note that in calculating the reward of each solution, the edges coming from the inclusion constraints (the blue dotted edges) are also counted. 
These three matchings are
\begin{align*}
\mathcal{P}_1 &= \{(a_1,b_1), (a_2,b_2), (a_3,b_4), (a_4,b_3)\},\\
\mathcal{P}_2 &= \{(a_1,b_1), (a_2,b_3), (a_3,b_2), (a_4,b_4)\},\\
\mathcal{P}_3 &= \{(a_1,b_2), (a_2,b_1), (a_3,b_3), (a_4,b_4)\}.
\end{align*}
The list $\bm{p}$ is 
\[
\phantom{.}\bm{p} = (\mathcal{P}_1, \mathcal{P}_2, \mathcal{P}_3).
\]
We remove the first matching from this list and set the second solution with the largest reward to be the codeword $\bm{c}_1$ corresponding to the matching $\mathcal{P}_1$.

We then partition the set of remaining matchings in Cell~$1$ by considering the ordered enumeration for the edges in $\bm{c}_1$. The resulting cells, together with their inclusion and exclusion constraints as well as their weight functions are shown in Fig.~\ref{fig:example_step21}. The inclusion and exclusion constraints of Cell~$1$ are carried forward in forming Cell~$1.1$, Cell~$1.2$ and Cell~$1.3$.
\end{example}
%%%%%%%%%%%%%%%%%%%%%%%%%%%%%%%%%%%%%%%%%%%
\begin{figure}
\centering
\includegraphics[width=\columnwidth]{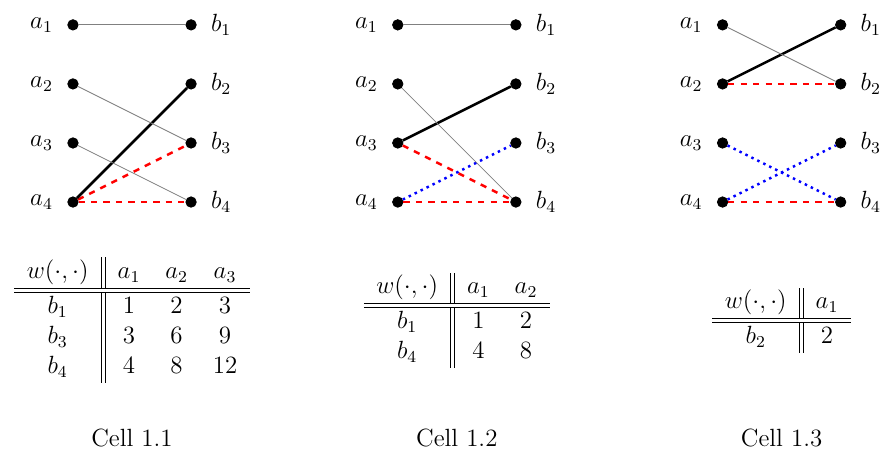}
\caption{The solution of the assignment problem for each cell is obtained after partitioning with respect to the ordered enumeration of $\bm{c}_1$.}
\label{fig:example_step22}
\end{figure}
%%%%%%%%%%%%%%%%%%%%%%%%%%%%%%%%%%%%%%%%%%%
The algorithm then solves the assignment problem in the newly partitioned cells. The solution of each cell is then added to the list of candidates $\bm{p}$. The next matching with the largest weight $\bm{c}_2$ will then be the first item in the sorted list $\bm{p}$. The algorithm proceeds by partitioning the corresponding cell with respect to the ordered enumeration of $\bm{c}_2$.
\begin{example}
The solutions of the assignment problem in Cell~$1.1$, Cell~$1.2$ and Cell~$1.3$ are shown in Fig.~\ref{fig:example_step22}. Again, the edge connected to vertex in $\mathcal{U}$ with the largest vertex weight is selected first. The weight function is correspondingly adjusted by removing one row and one column. After choosing the first edge, the rest of the edges in the maximum-weight matching are found by sorting (shown using thin gray lines). The solution in Cell~$1.1$ is referred to as $\mathcal{P}_{1.1}$ and has a weight of $27$. The solution in Cell~$1.2$ is called $\mathcal{P}_{1.2}$ and has a weight of $27$. The solution in Cell~$1.3$, $\mathcal{P}_{1.3}$, has a weight of $28$. The sorted list of candidate solutions is updated to
\[
\phantom{.}\bm{p} = (\mathcal{P}_2,\mathcal{P}_3, \mathcal{P}_{1.3}, \mathcal{P}_{1.1}, \mathcal{P}_{1.2}).
\]
The next solution in Murty's algorithm is then the codeword $\bm{c}_2$ corresponding to the matching $\mathcal{P}_2$.
The next step is to partition Cell~$2$ with respect to the ordered enumeration of $\bm{c}_2$. Continuing this, one can see that $\bm{c}_3$ is the codeword corresponding to the matching $\mathcal{P}_3$ and $c_4$ is the codeword corresponding to the matching $\mathcal{P}_{1.3}$.
\end{example}
This specialized Murty's algorithm continues until the $L$ solutions with the largest rewards are found.

\begin{remark}
In practice, we may have some extra constraints on the perfect matchings that disallow some matchings. That is, it is common to have an allowed set of perfect matchings $\mathcal{A}$ that is a proper subset of all possible perfect matchings. If one wishes to find the $L$ solutions to the assignment problem with the largest rewards subject to the extra constraint that the corresponding matchings should be in $\mathcal{A}$, once a candidate solution is chosen from the top of the list $\bm{p}$, the candidate solution should be tested for inclusion in the set of allowed solutions $\mathcal{A}$ and only then be considered as a member in the final list of $L$ solutions. If such a solution is disallowed because of non-membership of $\mathcal{A}$, it still must be kept on the sorted list $\bm{p}$ as the descendant cells may lead to a viable solution included in $\mathcal{A}$. We often expurgate permutation codes to have a subset with a power-of-two size, meaning that not all permutation codewords are allowed. This translates into a restriction on the set of allowed perfect matchings that we can have on the list of most likely codewords.
\end{remark} 
\vspace{-0.5cm}
\section*{Acknowledgment}
The authors wish to thank the reviewers of this paper for their insightful comments. Their constructive suggestions have significantly contributed to the improvement of this work. This work was supported in part by Huawei Technologies, Canada.

\vskip -2.5\baselineskip plus -1fil
\begin{IEEEbiographynophoto}{Reza Rafie Borujeny}
(Member, IEEE) received the B.Sc.\ degree from the University of
Tehran in 2012, the M.Sc.\ degree from the University
of Alberta in 2014 and the Ph.D.\ degree from the University
of Toronto in 2022 all in electrical and computer engineering.
His research interests include applications of information theory
and coding theory. 
\end{IEEEbiographynophoto}
\vskip -2\baselineskip plus -1fil
\begin{IEEEbiographynophoto}{Susanna~E.~Rumsey}
(Graduate Student Member, IEEE) was born in Toronto, ON, Canada in 1993.  She
received the B.A.Sc. degree with honours in engineering science (major in engineering
physics), and the M.Eng and M.A.Sc.~degrees in electrical and computer
engineering from the University of Toronto, Toronto, ON, Canada, in 2015, 2016,
and 2019 respectively.

Since 2019, she has been a Ph.D. student in electrical and computer engineering
at the University of Toronto, Toronto, ON, Canada.
\end{IEEEbiographynophoto}
\vskip -2\baselineskip plus -1fil
\begin{IEEEbiographynophoto}{Stark C. Draper}
(Senior Member, IEEE) received the B.S.\ degree in Electrical Engineering and 
the B.A.\ degree in History from Stanford University, and the M.S.\ and Ph.D.\ degrees 
in Electrical Engineering and Computer Science from the Massachusetts Institute 
of Technology (MIT). He completed postdocs at the University of Toronto (UofT) 
and at the University of California, Berkeley. He is a Professor in the Department 
of Electrical and Computer Engineering at the University of Toronto and was an 
Associate Professor at the University of Wisconsin, Madison. As a Research Scientist 
he has worked at the Mitsubishi Electric Research Labs (MERL), Disney's Boston 
Research Lab, Arraycomm Inc., the C. S. Draper Laboratory, and Ktaadn Inc. His 
research interests include information theory, optimization, error-correction 
coding, security, and the application of tools and perspectives from these fields 
in communications, computing, learning, and astronomy. He has been the recipient of 
the NSERC Discovery Award, the NSF CAREER Award, the 2010 MERL President's Award, 
and teaching awards from UofT, the University of Wisconsin, and MIT. He received an 
Intel Graduate Fellowship, Stanford's Frederick E. Terman Engineering Scholastic Award, 
and a U.S.\ State Department Fulbright Fellowship. He spent the 2019--2020 academic year 
on sabbatical visiting the Chinese University of Hong Kong, Shenzhen, and the 
Canada-France-Hawaii Telescope (CFHT), Hawaii, USA. Among his service roles, 
he was the founding chair of the Machine Intelligence major at UofT, 
was the Faculty of Applied Science and Engineering (FASE) representative on the UofT 
Governing Council, is the FASE Vice-Dean of Research, and is the President of 
the IEEE Information Theory Society for 2024.
\end{IEEEbiographynophoto}
\vskip -2\baselineskip plus -1fil
\begin{IEEEbiographynophoto}{Frank R. Kschischang}
(Fellow, IEEE) received the B.A.Sc.\ degree
(with honors) from the University of British Columbia, Vancouver, BC,
Canada, in 1985 and the M.A.Sc.\ and Ph.D.\ degrees from the University of
Toronto, Toronto, ON, Canada, in 1988 and 1991, respectively, all in
electrical engineering. Since 1991 he has been a faculty member in
Electrical and Computer Engineering at the University of Toronto, where he
presently holds the title of Distinguished Professor of Digital
Communication.

His research interests are focused primarily on the area of channel coding
techniques, applied to wireline, wireless and optical communication systems
and networks. He has received several awards for teaching and research,
including the 2010 Communications Society and Information Theory Society
Joint Paper Award, and the 2018 IEEE Information Theory Society Paper
Award.  He is a Fellow of IEEE, of the Engineering Institute of Canada, of
the Canadian Academy of Engineering, and of the Royal Society of Canada.

During 1997--2000, he served as an Associate Editor for Coding Theory for
the \textsc{IEEE Transactions on Information Theory}, and from 2014--16, he
served as this journal's Editor-in-Chief.  He served as general co-chair
for the 2008 IEEE International Symposium on Information Theory, and he
served as the 2010 President of the IEEE Information Theory Society. He
received the Society's Aaron D. Wyner Distinguished Service Award in 2016. 
He was awarded the IEEE Richard W.\ Hamming Medal for contributions to the 
theory and practice of error correcting codes and optical communications in 2023.
\end{IEEEbiographynophoto}
\end{document}